\documentclass[a4paper,titlepage]{article}

\title{Population Structures with Positive Feedback and Asymmetric Division 
\thanks{Submitted  {\textcolor{red}{April 12, 2026}}.}}
\date{April 12, 2026}

\author{Gabriel Dooley\thanks{Mathematics, Ohio University, Athens, OH 45701, gd957015@ohio.edu}
\and Camden Kilton\thanks{Mathematics and Physics \& Astrophysics, Ohio University, ck003122@ohio.edu} 
\and Brynley Needham\thanks{Mathematics, Ohio University, bn924722@ohio.edu}    
\and Graham Walther\thanks{Mathematics, Ohio University, gw514022@ohio.edu}
\and Todd R.\ Young\thanks{Mathematics, Ohio University, youngt@ohio.edu}
}

\usepackage{enumitem}
\setlist[enumerate]{leftmargin=.5in}
\setlist[itemize]{leftmargin=.5in}

\usepackage[mathlines]{lineno}

\usepackage{amsthm}
\newtheorem{theorem}{Theorem}[section]   
\newtheorem{lemma}[theorem]{Lemma}       

\newtheorem{corollary}[theorem]{Corollary}

\theoremstyle{definition}
\newtheorem{definition}[theorem]{Definition}

\newenvironment{keywords}{
    \par\vspace{1em}%
    \noindent\textbf{Keywords:} \itshape
}{\par\vspace{1em}}

\usepackage{amsmath}

\usepackage{cleveref}
\usepackage{amsfonts}

\usepackage{graphicx}
\graphicspath{{./figures/}}
\graphicspath{{figures/}} 

\usepackage{amsopn}

\usepackage{caption}

\usepackage[square,numbers]{natbib}

\usepackage{listings}
\usepackage{tikz}
\usetikzlibrary{decorations.markings}
\usetikzlibrary{calc,math}

\usepackage{acronym}
\newacro{PCS}{Perfectly Clustered Solution}
\newacro{CDC}{Cell Division Cycle}
\newacro{YMO}{Yeast Metabolic Oscillations}

\newcommand{\xb}{\mathbf{x}}
\newcommand{\Xb}{\mathbf{X}}

\begin{document}

\maketitle

\begin{abstract}
    In bioreactor experiments, budding yeast can manifest stable metabolic oscillations. 
    In some instances these oscillations involve the cell cycle 
    and are associated with a dynamical phenomenon called temporal clustering. 
     Since yeast divide asymmetrically, mother cells may be able to divide sooner 
     than the smaller daughter cells. 
     We show that asymmetric division plus positive feedback   
    can result in stable temporal clustering. 
    In this scenario, the cells self-organize into $p$ clusters of mother cells 
    and $q$ clusters of daughter cells, $p \le q$. 
    In simple numerical simulations of a population model with asymmetric division
    and positive feedback, we show that $p:q$ periodic arrangements form spontaneously 
    from random populations of cells. Our main result is that these structures
    can be stable.
    \\
\end{abstract}


\tableofcontents


\section{Introduction}



Temporal clustering is a phenomenon in which a
population of oscillators, cells, or other agents
become synchronized (or nearly so) in cohorts, with different
cohorts out of phase with each other (see \cite{orosz07,mauroy08,fernandez14}).  We distinguish
this from {\em synchrony} in which the phases of
the entire population are the same.
Although synchrony is common in nature and has been widely studied \cite{strogatz},
very few research articles have considered temporal clustering. 
Our main goal is to fill this gap in an application where it is relevant. 
A snapshot of a temporally clustered solution in the model we will
study is shown in \cref{fig:23plot}.
\begin{figure}[h!] 
\begin{center}
  \includegraphics[width=0.49\textwidth, height=4.0cm]{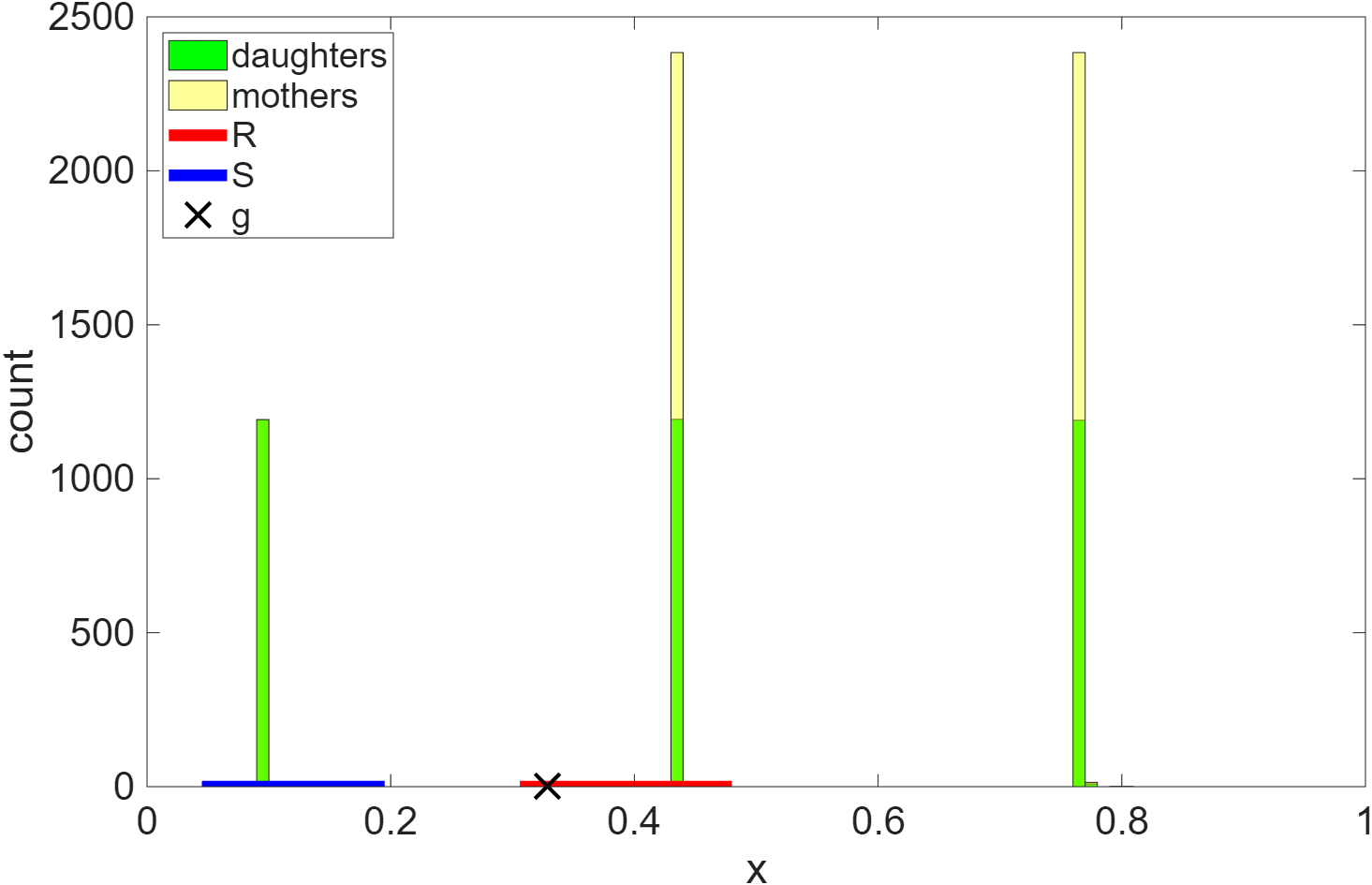} 
  \includegraphics[width=0.49\textwidth, height=4.0cm]{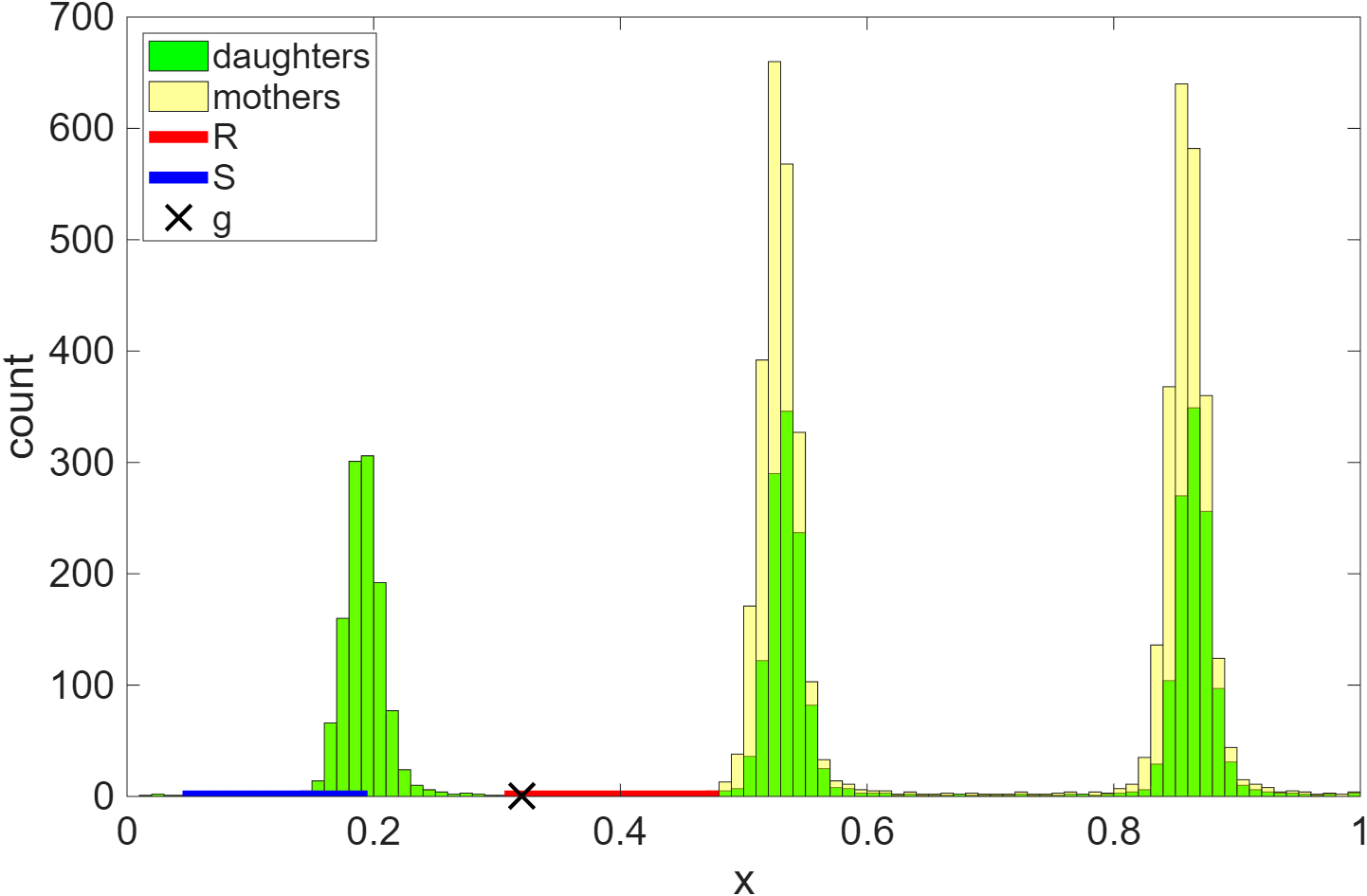}
\end{center}
\caption{
Under certain parameters in a model with asymmetric division and positive feedback, temporal clusters of cells arise from random  initial conditions.
The shown simulations used the following common parameters:
$r_1 = 0.306, $  $r_2 = 0.48,$ $s_1 = 0.045,$ $s_2 = 0.195$, and $\rho(I) = 0.6$. The tightly clustered simulation on the left used $g = 0.3288124$, while the right simulation used $g = 0.32$.
}
\label{fig:23plot}
\end{figure}

Temporal clustering has been shown to exist in `cell cycle-related' autonomous 
oscillations in bioreactor experiments \citep{hls11}, 
where metabolic variables such as dissolved oxygen and glucose
enter dramatic and stable oscillations with long periods (e.g., 4 hours).
See \cite{burnetti16} for more about yeast autonomous oscillations.

\citet{jbd10} proposed and analyzed models of populations of yeast in the 
cell cycle and showed that nonlocal, negative feedback can lead to stable temporal clustering.
Subsequently, temporal clustering was shown to be robust in different
versions of the model \cite{jtb12,gap14,stoch14,mediated14,metab18}. 
In these papers, symmetric division was a modeling assumption 
although it is well known that budding yeast divide asymmetrically \citep{lord,brewer84}; 
mother cells are larger than daughter cells immediately after division. 

Early in the study of autonomous oscillations, 
\citet{bellgardt94_I} proposed that asymmetric division 
might play a role in some cell cycle-related oscillations 
and lead to population structures in which cell density distributions 
consist of $p$ repeating patterns for mother cells
and $q$ repeating patterns for daughter cells, $p \le q$. 
See \cref{fig:PCSpqcell}.
He did not specifically propose temporal clustering, 
but that phenomenon is the simplest example of the suggested structure. 
Recently, \citet{naser25}, based on \cite{naser23}, proved that asymmetric division  
and {\em negative feedback} 
could produce asymptotically stable temporal clustering. 
In the current manuscript, we prove similar results under {\em positive feedback}. 
This is significant since Young et al.\ \cite{jtb12} proved 
that under symmetric division, positive feedback can only induce stability of
the {\em synchronized} solution, i.e.\ a single temporal cluster.

We will define and study a special type of $p:q$ periodic solution that we call a \acf{PCS}; 
see \cref{fig:PCSpqcell}.

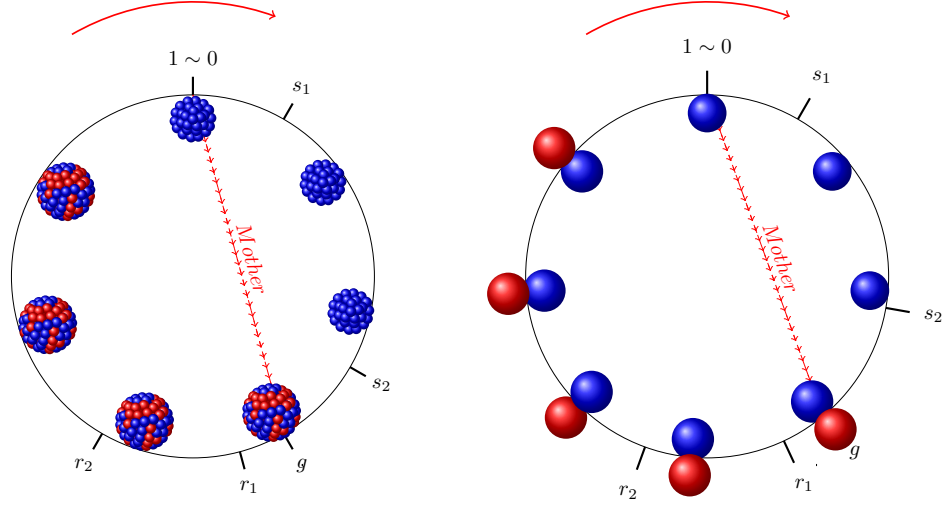
\begin{figure}[ht!]
\begin{center}  
\scalebox{.8}{
\begin{tikzpicture}
\begin{scope}
\draw (0,0) circle (3cm);
\foreach \angle / \label in 
{ 90/$1 \sim 0$,-60 /$ g$ , -75/$r_1$ ,-120/$r_2$ , 60/$s_1$, -30/$s_2$}
{
\draw[line width=1pt] (\angle:3cm) -- (\angle:3.3cm);
\draw (\angle:3.6cm) node{\textsf{\label}};
}
\foreach \angle in {-60}
\draw[line width=1pt] (\angle:3.6cm) -- (\angle:3.6cm);
\usetikzlibrary{calc}
\tikzset{
pics/proton/.style={code={\shade[ball color=red] circle (2.5pt);}},
pics/neutron/.style={code={\shade[ball color=blue] circle (2.5pt);}},
pics/nucleussmall/.style={code={%
	\pgfmathdeclarerandomlist{nucleon}{{neutron}{neutron}{neutron}{neutron}{neutron}}
	\pgfmathsetseed{#1+1}
	\foreach \A/\R in {14/0.3,8/0.2, 5/0.13, 1/0}{
		\pgfmathsetmacro{\S}{360/\A}
		\foreach \B in {0,\S,...,360}{
			\pgfmathrandomitem{\C}{nucleon}
			\pic at ($(\B+2*\A+5*rnd:\R)$) {\C}; } }} },
pics/nucleusbig/.style={code={%
	\pgfmathdeclarerandomlist{neutron}{{neutron}{neutron}{neutron}{neutron}{neutron}}
	\pgfmathsetseed{#1+1}
	\foreach \A/\R in {8/0.2, 5/0.13, 1/0}{
		\pgfmathsetmacro{\S}{360/\A}
		\foreach \B in {0,\S,...,360}{
			\pgfmathrandomitem{\C}{nucleon}
			\pic at ($(\B+2*\A+5*rnd:\R)$) {\C}; } }} },
pics/nucleusbiggest/.style={code={%
	\pgfmathdeclarerandomlist{neutron}{{neutron}{neutron}{proton}{neutron}{neutron}}
	\pgfmathsetseed{#1+1}
	\foreach \A/\R in {8/0.2, 8/0.13, 1/0}{
		\pgfmathsetmacro{\S}{360/\A}
		\foreach \B in {0,\S,...,360}{
			\pgfmathrandomitem{\C}{nucleon}
			\pic at ($(\B+2*\A+5*rnd:\R)$) {\C}; } }} },
pics/nucleusbig/.style={code={%
	\pgfmathdeclarerandomlist{nucleon}{{proton}{proton}{neutron}{neutron}{neutron}}
	\pgfmathsetseed{#1+1}
	\foreach \A/\R in {24/0.4, 24/0.3, 24/0.2, 13/0.35, 11/0.27, 6/0.15, 1/0}{
		\pgfmathsetmacro{\S}{360/\A}
		\foreach \B in {0,\S,...,360}{
			\pgfmathrandomitem{\C}{nucleon}
			\pic at ($(\B+2*\A+5*rnd:\R)$) {\C}; } }} }
}

\draw [red, dash pattern=on 3pt off 2pt,postaction={decorate,
		decoration={markings,
			mark=between positions 3pt and 1 step 5pt with {\arrow{>};}}}] (0,3) -- ++(1.5,-5.55);
		\foreach \angle in {-60}
		\draw[line width=1pt] (\angle:3.6cm) -- (\angle:3.6cm);

\draw[red,dashed] [-{latex[scale=5.0]}](0,3) -- (1.5,-2.55) node[pos=0.5,sloped,above] {$Mother$};

\pic at (0,2.6) {nucleussmall}; 
\pic at (2.15,1.55) {nucleussmall};
\pic at (2.61,-.58) {nucleussmall};
\pic at (1.3,-2.25) {nucleusbig=1}; 
\pic at (-.8,-2.4) {nucleusbig=1};
\pic at (-2.4,-.8) {nucleusbig=1};
\pic at (-2.1,1.4) {nucleusbig=1}; 
\draw[thick][red,->] (-2,4) arc (120:70:4cm) ;
\end{scope}
\begin{scope}[xshift=8.5cm]
\definecolor{qqqqff}{rgb}{0.,0.,1.}
\definecolor{ffqqqq}{rgb}{1.,0.,0.}
\draw (0,0) circle (3cm);

\draw [red, dash pattern=on 3pt off 2pt,postaction={decorate,
		decoration={markings,
			mark=between positions 3pt and 1 step 5pt with {\arrow{>};}}}] (0,3) -- ++(2,-5.55);
		\foreach \angle in {-60}
		\draw[line width=1pt] (\angle:3.6cm) -- (\angle:3.6cm);

\draw[red,dashed] [-{latex[scale=5.0]}](0,3) -- (2,-2.55) node[pos=0.5,sloped,above] {$Mother$};

\foreach \angle / \label in 
{ 90/$1 \sim 0$,-50 /$ g$ , -65/$r_1$ ,-110/$r_2$ , 60/$s_1$, -10/$s_2$}
{
\draw[line width=1pt] (\angle:3cm) -- (\angle:3.4cm);
\draw (\angle:3.8cm) node{\textsf{\label}};
}
\foreach \angle in {-50}
\draw[line width=1pt] (\angle:3cm) -- (\angle:3.6cm);

\shade[ball color=qqqqff] (265:2.7) circle (3.5mm);
\shade[ball color=ffqqqq] (265:3.3)  circle (3.5mm);

\shade[ball color=qqqqff] (-50:2.7) circle (3.5mm);
\shade[ball color=ffqqqq] (-50:3.3)  circle (3.5mm);

\shade[ball color=qqqqff] (140:2.7) circle (3.5mm);
\shade[ball color=ffqqqq] (140:3.3)  circle (3.5mm);

\shade[ball color=qqqqff] (185:2.7) circle (3.5mm);
\shade[ball color=ffqqqq] (185:3.3)  circle (3.5mm);

\shade[ball color=ffqqqq] (225:3.3) circle (3.5mm);
\shade[ball color=qqqqff] (225:2.7) circle (3.5mm);

\shade[ball color=qqqqff] (90:2.7) circle (3.2mm);
\shade[ball color=qqqqff] (-5:2.7) circle (3.2mm);
\shade[ball color=qqqqff] (40:2.7) circle (3.2mm);
\draw[thick][red,->] (-2,4) arc (120:70:4cm) ;
\end{scope}
\end{tikzpicture}
}  
\captionof{figure}{Left: Structure of a $4:7$ \ac{PCS} where blue balls express daughter cells and red balls express mother cells.  
Right: A $5:8$ \ac{PCS} with the clustered models (blue daughter clusters 
and red mother clusters). 
Observe that mother cells jump from $1 \sim 0$ to $g$ as depicted by the red arrows, 
while daughter cells progress continuously at $1 \sim 0$.}
\label{fig:PCSpqcell}
\end{center}
\end{figure}
These solutions are very special in that both the existence and the stability necessitate restrictive conditions on the model parameters.
However,  we show that these conditions can be satisfied. 
The asymptotic stability then implies that such a \ac{PCS} 
is also structurally stable, 
i.e.\ similar stable solutions exist for an open set of  nearby parameter values.

In \cref{sec:cells} we introduce models that describe the dynamics of individual cells. 
In simulations of the model, such as \cref{fig:23plot}, $p:q$ solutions form spontaneously.
A key mathematical challenge is that the simplest model is
not continuous; mother cells are modeled as jumping ahead in
the cycle; see \cref{fig:PCSpqcell}.
\citet{naser23} overcame this difficulty by introducing complementary
smooth models called Fixed and Switching that we review in \cref{sec:clusters}.
In \cref{sec:perfect} we formally introduce \ac{PCS}s, 
an ideal version of $p:q$ periodic solution. 
In \cref{sec:returnmaps}, we introduce special return maps, $F_f$ and $F_s$, 
that can be viewed as factor maps of the Poincar\'e return map \cite{jtb12}.  

In \cref{sec:12,sec:23}, we will prove that mode $1:2$ and $2:3$ \ac{PCS} can exist 
and be linearly asymptotically stable 
by showing that the eigenvalues of maps $F_s$ and $F_f$ 
are in the interior of the unit circle in $\mathbb{C}$. 
By linearizing the maps \(F_\alpha\) at the fixed point, 
we can use the chain rule to determine the stability of the fixed point under \(P_\alpha\),
corresponding to the periodic \ac{PCS}.
A crucial ingredient in the analysis is the calculation of exact solutions 
for short periods of time to form the return maps.

\section{Models and Dynamical Systems Methods}

\subsection{The Cell Cycle Model} 
\label{sec:cells}

The cell cycle is a biological process in which a cell matures and divides into 2 cells. 
It consists of many complex chemical and physical processes that are known to be sensitive
to agents and conditions in the culture (see e.g.\ \cite{beuse98,beuse99,burnetti16}). 
In the model proposed in \cite{jtb12} to capture the essense of this sensitivity, the cell cycle is represented in coordinates $[0,1]$
with $1$ denoting division. A subinterval \( R = [r_1, r_2] \subset [0,1] \) is termed the \textit{responsive region},
and \( S = [s_1, s_2] \) is the \textit{signaling region}, where we assume 
\( 0 < s_1 < s_2 < r_1 < r_2 < 1 \). 
Cells in \( S \) affect the speed of progress of cells in \( R \) 
via a feedback function \( \rho \). 
The position of a cell, 
denoted by \(x_i(t) \in [0, 1]\), 
matures over the interval \([0, 1)\) according to the following equation:
\begin{equation}\label{basiccell}
\frac{dx_i}{dt}=
\begin{cases}
   1, \quad \textrm{if} \quad x_i \notin R \\
1+  \rho(I)  , \quad \textrm{if} \quad x_i \in R
\end{cases}
\end{equation}
for each cell \( i=1, \ldots, n\). 
Here, \( I \) represents the proportion of cells located within the signaling region \( S \), that is, 
\begin{equation}\label{I}
I \equiv \frac{\#\{i : x_i \in S\}}{n}. 
\end{equation}
In this work we consider positive feedback, which means \( \rho \) is a monotonically increasing function with $\rho (0) = 0$.

In the original model \cite{jtb12} and similar studies, all cells were treated as identical. 
However, in this work, 
we will categorize cells as either \textit{mothers} or \textit{daughters}. 
When a cell reaches 1 it either goes to 0 in the cell cycle, 
becoming a \textit{daughter} cell, 
or goes to \( g  \in (0, 1)\), becoming a \textit{mother} cell. 
By starting at $g>0$, the mother cells are more mature and are able to reach the next division
more quickly than daughter cells, representing assymetric division.
See \cref{fig:PCSpqcell}.

		
	

In \ac{YMO} experiments (for instance \cite{palkova00}), 
cells are harvested from the bioreactor at a constant rate, 
keeping the number of cells near the carrying capacity. 
Therefore, it is reasonable to hold the number of cells fixed at \( n\). 
At division of budding yeast, exactly half of the cells produced are mothers and half are daughters. 
To reflect this biological fact in a model, we will assign cells the roles of either  \textit{mother} or \textit{daughter} in an alternating pattern.
\begin{definition}
Model $\bf c_a$. 
As cells reach $1$, they {\em alternate} - they are assigned the opposite designation of the 
previous cell. 
\end{definition}
In other words, the designation of cells at division alternates between \textit{mother} (go to \(g\)) 
or \textit{daughter} (go to \(0\)). 
Together, \cref{basiccell,I}, as well as $\bf{c_a}$ 
make up the \textit{Alternating Cell Model}.

Implementing the $\bf{c_a}$ model in a Matlab program, 
simulations of \cref{basiccell,I} 
were created to demonstrate Perfectly Clustered Solutions (\ac{PCS}) of modes  2:3 and 1:2 in \cref{fig:12plot,fig:23plot}.
The program assigns cells to random initial conditions and runs the model numerically. 
At the end of the simulation, the program 
records a histogram showing the distribution of cells in the cell cycle. 
Under specified parameters, clustered solutions appear spontaneously. 
This suggests that these solutions are asymptotically stable periodic states 
of the system. 
The program is included in the Supplementary Materials.

In \cref{fig:12plot}, 
there are 3000 cells that have self-organized into 2 daughter clusters and 1 mother cluster. 
Once formed, one daughter cluster overlaps and travels with the mother cluster. 
We observe a clearly clustered solution of mode 1:2 in the left image, 
as well as a \ac{PCS} of mode 1:2 in the right image. 
In order for the perfectly clustered structure to remain coherent, 
a daughter cluster must reach $g$ at the same time that the mother/daughter pair
reach $1$.

\begin{figure}[ht] 
  \includegraphics[width=0.49\textwidth, height=4.0cm]{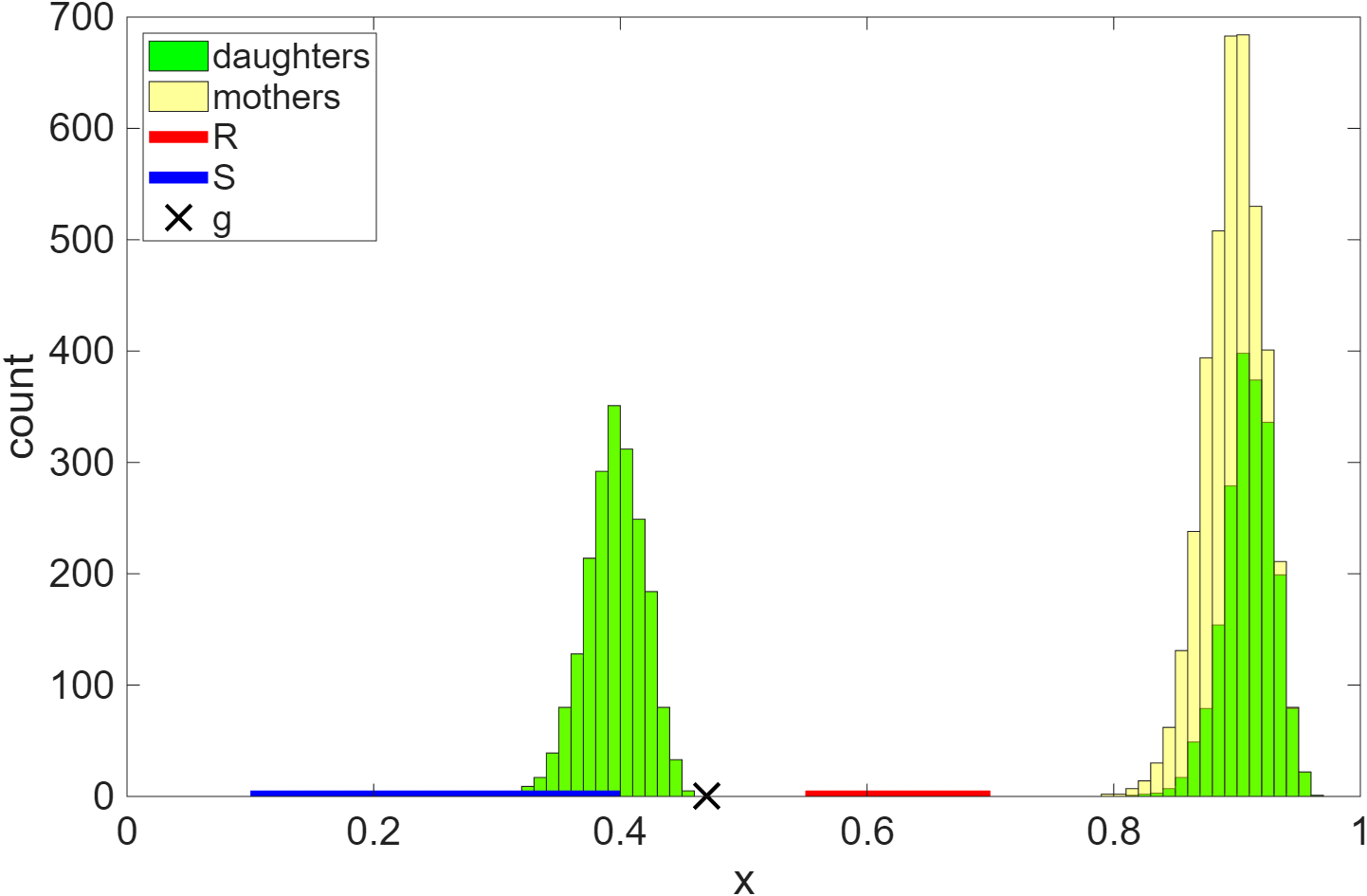} 
  \includegraphics[width=0.49\textwidth, height=4.0cm]{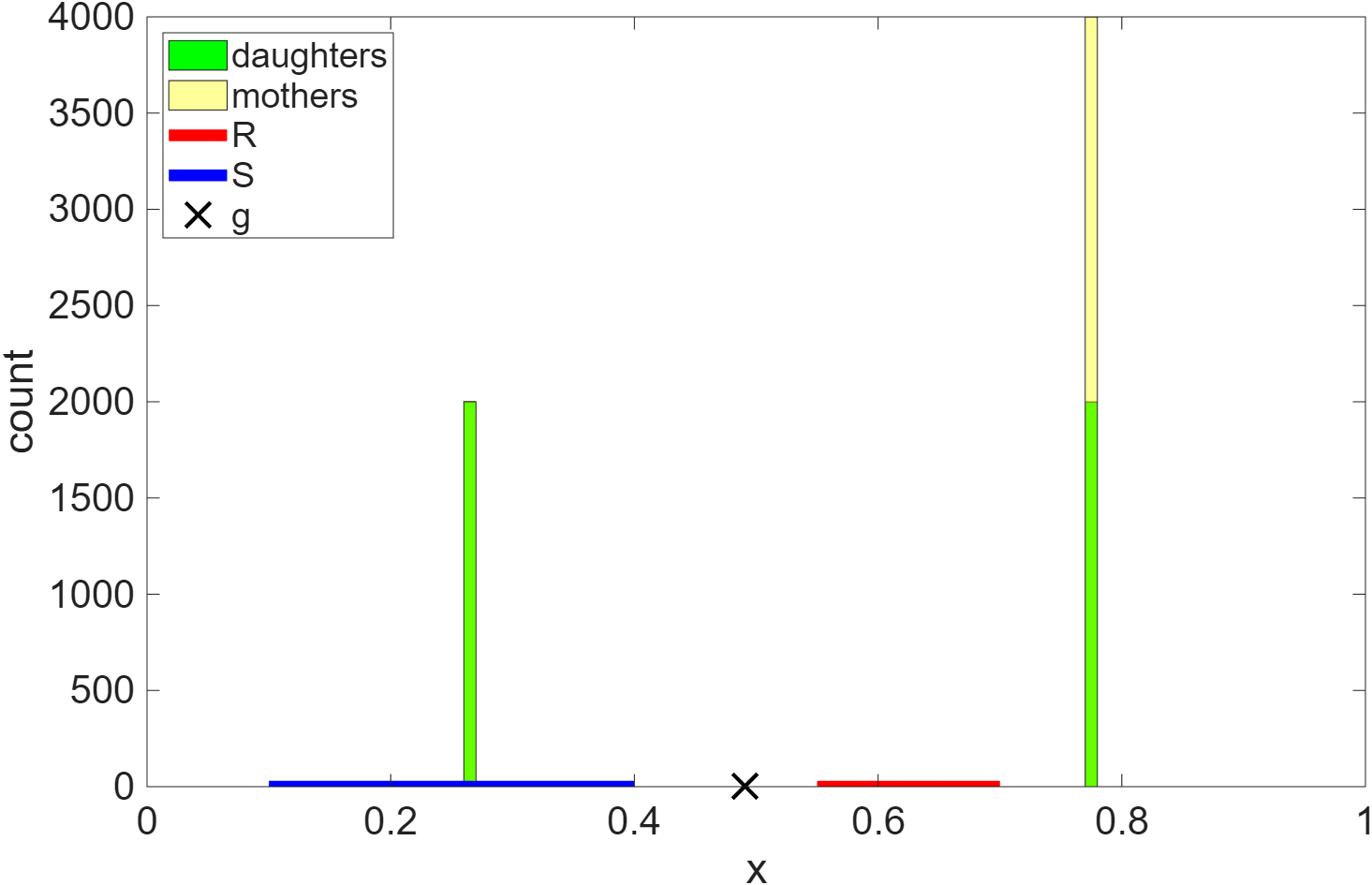}
\caption{Histogram of temporal clusters for \cref{basiccell,I} with convention $\bf c_a$. 
Temporal clusters  arise from simulations with random initial conditions of 3000 cells. 
The image on the left is a solution that is clearly clustered, 
but not perfect, where the parameters are as follows: 
$r_1 = 0.55, $  $r_2 = 0.7,$ $s_1 = 0.1,$ $s_2 = 0.4$, $\rho(I) = 0.6$ and $g = 0.47$. 
The smaller cluster consists of exactly 1000 daughters while the two large clusters are made up of 1000 daughters and 1000 mothers. 
The image on the right has the same parameters, except for $g$, which is changed to $g=0.4911$, resulting in a Perfectly Clustered Solution. 
} 
\label{fig:12plot}
\end{figure}

Similarly, in \cref{fig:23plot}, we produced histograms for a snapshot of mode 2:3 solutions. 
We initialize a total of 5000 randomly placed cells. 
At any given time, one daughter cluster is moving independently, 
while each of the other two daughter clusters are moving with a mother cluster. 
On the left is a clearly clustered solution of mode 2:3, 
and on the right is a near \ac{PCS} of mode 2:3.
We note that in \cref{fig:12plot,fig:23plot}, each cluster happens to have exactly 1000 cells.
This equidistribution of cells in temporal clusters seems to be a universal phenomenon 
\cite{siads20}.

While $p:q$ simulations can be created and \ac{PCS}s can be observed in higher values of both $p$ and $q$, 
the biological significance of these results would be negligible. 
As the number of clusters grows, the number of cells within each cluster diminishes, 
reducing the strength of feedback signals between cells. 
This weakened feedback should eventually become insufficient to coordinate cellular progression
in the presence of biological noise.

This model with assumption $\bf c_a$ 
does not generate a continuous dynamical system, since mother cells jump from $1$ to $g$. 
Therefore, following \cite{naser23,naser25}, we will also consider the following conventions that do in fact generate continuous systems on appropriate phase spaces:
\begin{definition} 
\label{def:cfs}
Consider \cref{basiccell,I} with the assumptions:
\begin{itemize}
\item[$\bf c_f$] - individual cells keep their mother/daughter identity when they cross $1$ and 
\item[$\bf c_s$] - cells switch their mother/daughter identity when they cross $1$.  
\end{itemize}
We call \cref{basiccell,I} with $\bf c_f$ the {\em fixed model} and with $\bf c_s$ the {\em switching model}.
\end{definition}

We note that in order to achieve a \ac{PCS} in models $\bf c_f$ and $\bf c_s$, 
we must begin with the correct proportions of daughter and mother cells. 
For the $1:2$ case, this means that we must begin with $\frac{2}{3}$ of our cells being daughters 
and $\frac{1}{3}$ of our cells being mothers.


\subsection{The Cluster Model and Perfectly Clustered Solutions} 
\label{sec:perfect}
\label{sec:clusters}

As in \cite{naser25}, we will use ``cluster models" to simplify calculations and streamline analysis.
Rather than tracking the positions of thousands of cells, 
we treat synchronized cells as one entity, which we call a cluster.
Two cells are synchronized in $[0, 1)$, if their positions are the same, 
at least for short periods of time, e.g.\ until they reach division.

We let $D_i \in [0, 1)$ denote the position of a cluster containing daughter cells,
and $M_j \in [g, 1)$ denote the position of a cluster containing mother cells.
If $D_i=M_j$, they are still considered as separate entities. 
The clusters' positions, following from \cref{basiccell}, evolve by the following equations
\begin{equation}\label{eqn:clusters}
\frac{d(D_i,M_i)}{dt}  = \begin{cases}
1, \quad \textrm{if} \quad D_i, M_i \notin R \\
1+  \rho(I)  , \quad \textrm{if} \quad D_i, M_i\in R 
\end{cases}
\end{equation}
where $R$ is  the responsive region as before, and $I$ is as stated in \cref{I}.
Therefore, instead of equations for each individual cell, 
in the case of a $p:q$ clustered solution we have equations for 
$M_1(t),$ $M_2(t)$ $...M_p(t)$ and $D_1(t),$ $ D_2(t)$ $...D_q(t)$. 
In the cluster model, with each cluster having the same number of cells, 
it follows that the
previously defined $I$ \cref{I} equals the fraction of clusters in $S$.
To avoid the dispersion of the clusters as they pass $1$, 
we will also follow the same conventions from \cite{naser25} that mirror \cref{def:cfs}.
\begin{definition}
For a system of $p$ mother clusters and $q$ daughter clusters, define:
\begin{enumerate}
\item[$\bf C_f$.]   Daughter clusters remain daughter clusters  (i.e.\ go to $0$) and mother clusters remain mothers (go to $g$) when they reach $1$. We will call this  the   {\em  fixed clustered  model}.
\item[$\bf C_s$.]    Daughter clusters become mother clusters (go to $g$) and mother clusters become  daughters (go to $0$) when they reach $1$.   We call this the {\em switching clustered model.}
\end{enumerate}
\end{definition}
For any solution $\Xb(t)$ in $\bf C_f$ or $\bf C_s$ there is a clustered solution $\xb(t)$ 
in $\bf c_f$ or $\bf c_s$ whose coordinates are given by $\Xb(t)$ \cite{naser25}. Thus $\bf C_f$
can be considered as an invariant subspace of $\bf c_f$ and similarly for $\bf C_s$ and $\bf c_s$.

A \ac{PCS} is able to occur in any of the $\bf c_a$, $\bf c_f$, $\bf c_s$, $\bf C_f$,  or $\bf C_s$ models.
In order for our models to exhibit a \ac{PCS}, the progression of cell groups must align exactly.
If a cluster of mother cells reach $1$ and advance to $g$, a cluster of 
daughter cells must also reach $g$ at that exact moment. 
We restrict our attention to cases where such timing is perfectly coordinated, 
ensuring that the groups stay in perfect step as they move through the cycle.
\begin{definition}
\label{perfect}
We call a solution a {\em mode $p:q$  \acf{PCS}} if for all $t\ge 0$ (a) all cells are in
$q$ daughter clusters or $p$ mother clusters, (b) each mother or daughter cluster contains the same
number of cells ($n/(p+q)$) (c) each mother cluster is always exactly synchronized
with a daughter cluster. 
\end{definition}

A \ac{PCS} in the cell model matches precisely to a \ac{PCS} in the cluster model.
We illustrate \ac{PCS}s in \cref{fig:PCSpqcell} 
using cells (left) and clusters (right).
A slight modification of a proof in \cite{naser25} gives the following.
\begin{theorem}
\label{thm:existence}
   For any feedback function (positive or negative), any set of parameters \(0<s_1<s_2<r_1<r_2<1 \) and any \(1\leq p\leq q\), there is a \(g\in[0,1)\) such that a \(p:q\) PCS exists.
\end{theorem}

A PCS represents a highly structured solution. 
For a $p:q$ configuration (with $p\leq q$) to exist,
the total number of cells must be divisible by $p+q$. 
There is always a trivial \ac{PCS} in the absence of any 
feedback under specific timing conditions.
However, this configuration could be easily disrupted; 
even a slight change in $g$ breaks the periodicity. 
We show in this paper that positive feedback can stabilize a \ac{PCS}.


\subsection{Return Maps for Clustered Models} 
\label{sec:returnmaps}

Let \(\Xb(t)\) be a mode \(p:q\) PCS with minimal period \(T\). 
Define \(d_i = D_i(0)\) and \(m_i = M_i(0)\). 
Since the time derivative of $M_1(t)$ is always positive by assumption, 
it follows that every solution intersects the set $\Sigma \equiv \{ M_1 = g \}$ in finite positive time.
In dynamical systems theory, $\Sigma$ is called a {\em Poincar\'{e} section}.  
For a PCS, if \(m_1 = g\), then \(d_1 = 0\) and \(M_1(T) = g\), 
where $T$ is the period of the PCS.
Let \(\hat{\Xb}(0)\) be an initial condition that lies on \(\Sigma\) 
in a small neighborhood of \(\Xb(0)\). We now define a {\em Poincar\'{e} first return
map} on this small neighborhood, 
which simply describes how solutions starting from these initial
conditions on $\Sigma$ return to $\Sigma$.
For each $\hat{\Xb}(0)$, let \(\tau \approx T\) be such that \(\hat{\Xb}(\tau) \in \Sigma\). 
For the models \(\bf C_f\) or $\bf C_s$, 
define the Poincaré map \(P_\alpha\), $\alpha = \mathbf{f}, \mathbf{s}$ by
\begin{equation*}
\begin{split}
P_\alpha: \ & (d_1, ..., d_q, m_1=g, m_2, ..., m_p)  \\
       & \hspace*{1cm}     \mapsto   (D_1(\tau),...,D_q(\tau),M_1(\tau)=g,M_2(\tau),...,M_p(\tau))
 \end{split}
\end{equation*}
As in \cite{naser25}, the Poincaré maps are piecewise affine, and the stability of the PCS is determined by the eigenvalues of the derivative (Jacobian) \(DP_\alpha\). If \(\max_{\lambda \in \sigma(DP_\alpha)} |\lambda| < 1\), the PCS is locally asymptotically stable. If \(\max_{\lambda \in \sigma(DP_\alpha)} |\lambda| = 1\), the PCS is at most neutrally stable. If \(\max_{\lambda \in \sigma(DP_\alpha)} |\lambda| > 1\), the PCS is unstable.

For the fixed cluster model \(\mathbf{C_f}\), define \(t^*_\mathbf{f}\) by \(M_p(t^*_\mathbf{f}) = 1\), and define a map \(F_\mathbf{f}\) on a small neighborhood \(U \subset \Sigma\) of \(\Xb(0)\) by \[
F_\mathbf{f} : (d_1, d_2,..., d_q, m_2,..., m_p) \mapsto (D_q(t^*_\mathbf{f}), D_1(t^*_\mathbf{f}),..., D_{q-1}(t^*_\mathbf{f}), M_1(t^*_\mathbf{f}),..., M_{p-1}(t^*_\mathbf{f}))
\]
where we note that \(m_1 = g\) and \(M_p(t^*_\mathbf{f}) = g\) 
and thus these coordinates can be omitted from the map. 
Each cluster of the perturbed solution is mapped by \(F\) 
to a small neighborhood of the initial position of the cluster ahead of it. 
For instance, \(D_{i-1}(t^*_\mathbf{f})\) is in a neighborhood of \(d_i\). 
If we identify a small neighborhood of the PCS on the surface \(\{ M_p = g \}\) with the neighborhood \(U \subset \Sigma\), we can simply consider \(F_\mathbf{f}\) to be a map from \(U\) to itself.

For the clustered switching model \(\mathbf{C}_\mathbf{s}\), 
define \(t^*_\mathbf{s}\) by \(D_q(t^*_\mathbf{s})=1\), and define 
\begin{equation*}
F_\mathbf{s}:(d_1,\dots d_q,m_2,...,m_p) \mapsto (M_p(t^*_s),D_1(t^*_s),...,D_{q-1}(t^*_s),M_1(t^*_s),...,M_{p-1}(t^*_s))
\end{equation*}
Note that at time \(t^*_\mathbf{s}\), \(D_q\) becomes a mother with position \(g\), 
while \(M_p\), who was synchronized with \(D_q\), 
becomes a daughter when it crosses \(0\). 
We omit the coordinates \(m_1\) and \(D_q(t^*_\mathbf{s})\) 
from the map since they are both fixed at \(g\). 
As above, for a small neighborhood \(U \subset \Sigma\) of \(\Xb(0)\), 
we can consider \(F_s\) a map from \(U\) to itself.

In \cref{fig:F12}, we illustrate the actions of $F_f$, $F_s$ and $P$ on a 1:2 \ac{PCS}. 
We see that under the \(\mathbf{C_s}\) model, 
each cluster of a \ac{PCS} is mapped by \(F_\mathbf{s}\) to the initial position of the next cluster. 
Under iteration, it cycles through every initial position until returning to its own on the 
\((p+q)^{\text{th}}\) iteration. 
For \(C_\mathbf{f}\), all clusters return to their original positions after 
\(\operatorname{lcm}(p,q)\) iterations.

\begin{figure}[ht]
\begin{center} 
\scalebox{.8}{ 
\begin{tikzpicture}
\begin{scope} [xshift=-4cm,yshift=0cm]
\definecolor{qqqqff}{rgb}{0.,0.,1.}
\definecolor{ffqqqq}{rgb}{1.,0.,0.}
\definecolor{Yellow}{rgb}{1, 1, 0}
\draw (0,0) circle (1.25cm);
\foreach \angle / \label in 
{ 90/$1 \sim 0$,-80 /$g
$ , -100/$r_1$ ,-140/$r_2$ , 60/$s_1$, -30/$s_2$}
{
\draw[line width=1pt] (\angle:1.5cm) -- (\angle:1.25cm);
\draw (\angle:1.9cm) node{\textsf{\label}};
}
\foreach \angle in {-80}
\draw[line width=1pt] (\angle:1.6cm) -- (\angle:1.2cm);
\shade[ball color=ffqqqq] (-80:1.4) circle (1.5mm);
\shade[ball color=Yellow] (-80:1.1) circle (1.5mm);
\shade[ball color=qqqqff] (90:1.1) circle (1.5mm);
\end{scope}
\draw[thick][blue,->] (-2.5,0) --node[above]{$F_{\bf f}$}(-1.5,0);
\begin{scope} [xshift=0cm,yshift=0cm]
\definecolor{qqqqff}{rgb}{0.,0.,1.}
\definecolor{ffqqqq}{rgb}{1.,0.,0.}
\definecolor{Yellow}{rgb}{1, 1, 0}
\draw (0,0) circle (1.25cm);
\foreach \angle / \label in 
{ 90/$1 \sim 0$,-80 /$g
$ , -100/$r_1$ ,-140/$r_2$ , 60/$s_1$, -30/$s_2$}
{
\draw[line width=1pt] (\angle:1.5cm) -- (\angle:1.25cm);
\draw (\angle:1.9cm) node{\textsf{\label}};
}
\foreach \angle in {-80}
\draw[line width=1pt] (\angle:1.6cm) -- (\angle:1.2cm);
\shade[ball color=qqqqff] (-80:1.1) circle (1.5mm);
\shade[ball color=ffqqqq] (-80:1.4) circle (1.5mm);
\shade[ball color=Yellow] (90:1.1) circle (1.5mm);
\end{scope}
\draw[thick][blue,->] (1.5,0) --node[above]{$F_{\bf f}$}(2.5,0);
\begin{scope} [xshift=4cm,yshift=0cm]
\definecolor{qqqqff}{rgb}{0.,0.,1.}
\definecolor{ffqqqq}{rgb}{1.,0.,0.}
\definecolor{Yellow}{rgb}{1, 1, 0}
\draw (0,0) circle (1.25cm);
\foreach \angle / \label in 
{ 90/$1 \sim 0$,-80 /$g
$ , -100/$r_1$ ,-140/$r_2$ , 60/$s_1$, -30/$s_2$}
{
\draw[line width=1pt] (\angle:1.5cm) -- (\angle:1.25cm);
\draw (\angle:1.9cm) node{\textsf{\label}};
}
\foreach \angle in {-80}
\draw[line width=1pt] (\angle:1.6cm) -- (\angle:1.2cm);
\shade[ball color=ffqqqq] (-80:1.4) circle (1.5mm);
\shade[ball color=Yellow] (-80:1.1) circle (1.5mm);
\shade[ball color=qqqqff] (90:1.1) circle (1.5mm);
\end{scope}

\begin{scope} [xshift=-4cm,yshift=-5cm]
\definecolor{qqqqff}{rgb}{0.,0.,1.}
\definecolor{ffqqqq}{rgb}{1.,0.,0.}
\definecolor{Yellow}{rgb}{1, 1, 0}
\draw (0,0) circle (1.25cm);
\foreach \angle / \label in 
{ 90/$1 \sim 0$,-80 /$g
$ , -100/$r_1$ ,-140/$r_2$ , 60/$s_1$, -30/$s_2$}
{
\draw[line width=1pt] (\angle:1.5cm) -- (\angle:1.25cm);
\draw (\angle:1.9cm) node{\textsf{\label}};
}
\foreach \angle in {-80}
\draw[line width=1pt] (\angle:1.6cm) -- (\angle:1.2cm);
\shade[ball color=ffqqqq] (-80:1.4) circle (1.5mm);
\shade[ball color=Yellow] (-80:1.1) circle (1.5mm);
\shade[ball color=qqqqff] (90:1.1) circle (1.5mm);
\end{scope}
\draw[thick][blue,->] (-2.5,-5) --node[above]{$F_{\bf s}$}(-1.5,-5);
\begin{scope} [xshift=0cm,yshift=-5cm]
\definecolor{qqqqff}{rgb}{0.,0.,1.}
\definecolor{ffqqqq}{rgb}{1.,0.,0.}
\definecolor{Yellow}{rgb}{1, 1, 0}
\draw (0,0) circle (1.25cm);
\foreach \angle / \label in 
{ 90/$1 \sim 0$,-80 /$g
$ , -100/$r_1$ ,-140/$r_2$ , 60/$s_1$, -30/$s_2$}
{
\draw[line width=1pt] (\angle:1.5cm) -- (\angle:1.25cm);
\draw (\angle:1.9cm) node{\textsf{\label}};
}
\foreach \angle in {-80}
\draw[line width=1pt] (\angle:1.6cm) -- (\angle:1.2cm);
\shade[ball color=qqqqff] (-80:1.1) circle (1.5mm);
\shade[ball color=Yellow] (-80:1.4) circle (1.5mm);
\shade[ball color=ffqqqq] (90:1.1) circle (1.5mm);
\end{scope}
\draw[thick][blue,->] (1.5,-5) --node[above]{$F_{\bf s}$}(2.5,-5);
\begin{scope} [xshift=4cm,yshift=-5cm]
\definecolor{qqqqff}{rgb}{0.,0.,1.}
\definecolor{ffqqqq}{rgb}{1.,0.,0.}
\definecolor{Yellow}{rgb}{1, 1, 0}
\draw (0,0) circle (1.25cm);
\foreach \angle / \label in 
{ 90/$1 \sim 0$,-80 /$g
$ , -100/$r_1$ ,-140/$r_2$ , 60/$s_1$, -30/$s_2$}
{
\draw[line width=1pt] (\angle:1.5cm) -- (\angle:1.25cm);
\draw (\angle:1.9cm) node{\textsf{\label}};
}
\foreach \angle in {-80}
\draw[line width=1pt] (\angle:1.6cm) -- (\angle:1.2cm);
\shade[ball color=qqqqff] (-80:1.4) circle (1.5mm);
\shade[ball color=ffqqqq] (-80:1.1) circle (1.5mm);
\shade[ball color=Yellow] (90:1.1) circle (1.5mm);
\end{scope}
\draw[thick][blue,->] (5.5,-5) --node[above]{$F_{\bf s}$ }(6.5,-5);
\begin{scope} [xshift=8cm,yshift=-5cm]
\definecolor{qqqqff}{rgb}{0.,0.,1.}
\definecolor{ffqqqq}{rgb}{1.,0.,0.}
\definecolor{Yellow}{rgb}{1, 1, 0}
\draw (0,0) circle (1.25cm);
\foreach \angle / \label in 
{ 90/$1 \sim 0$,-80 /$g
$ , -100/$r_1$ ,-140/$r_2$ , 60/$s_1$, -30/$s_2$}
{
\draw[line width=1pt] (\angle:1.5cm) -- (\angle:1.25cm);
\draw (\angle:1.9cm) node{\textsf{\label}};
}
\foreach \angle in {-80}
\draw[line width=1pt] (\angle:1.6cm) -- (\angle:1.2cm);
\shade[ball color=ffqqqq] (-80:1.4) circle (1.5mm);
\shade[ball color=Yellow] (-80:1.1) circle (1.5mm);
\shade[ball color=qqqqff] (90:1.1) circle (1.5mm);
\end{scope}
\end{tikzpicture}
} 
\captionof{figure}{
The dynamics of a 1:2 \ac{PCS} under the fixed model $\bf C_f$ (top), 
and switching model $\bf C_s$ (bottom) 
applying the maps until we get back to the initial configuration, 
i.e.\ the action of the Poincar\'e map.
The interior balls represent clusters that are currently daughters, 
and the exterior balls represent current  mother clusters.}
\label{fig:F12}
\end{center}
\end{figure}
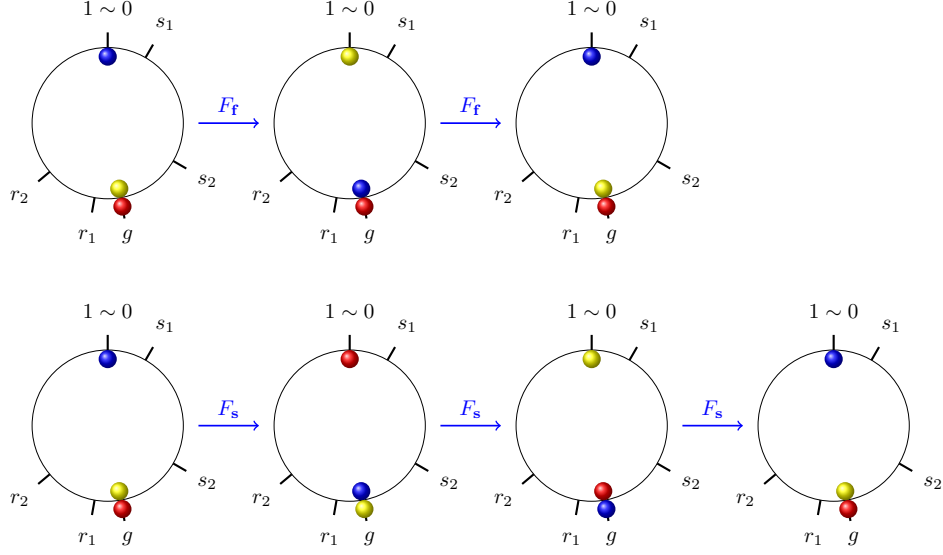


\subsection{Events and Orders}
\label{sec:events}

\begin{definition}
    Let \((D_1(t),D_2(t),...,D_q(t))\) be the coordinates of the daughter clusters in \([0,1)\). We define the following:
    \begin{itemize}
        \item \(D_i(t)=s_1\), for some \(i\), is called the {\em event \(\mathbf{s_1}\)}
        \item \(D_i(t)=s_2\), for some \(i\), is called the {\em event \(\mathbf{s_2}\)}
        \item \(D_i(t)=r_1\), for some \(i\), is called the {\em event \(\mathbf{r_1}\)}
        \item \(D_i(t)=r_2\), for some \(i\) is called the {\em event \(\mathbf{r_2}\)}
        \item \(D_i(t)=1\), for some \(i\) is called the {\em event \(\mathbf{1}\)}
    \end{itemize}
\end{definition}
For any solution, there is an associated sequence of events. For a periodic solution, 
the sequence of events is necessarily periodic.  
We will assume for the rest of the paper that no two events occur simultaneously. 
Thus, if \(\Xb(0)\) is the initial condition of a PCS 
and \(\hat{\Xb}(0)\) is a sufficiently close initial condition of another solution, 
then \(\Xb(t)\) and \(\hat{\Xb}(t)\) will share an order of events, 
at least during a finite interval of time.

\begin{definition}
    Let \(\Xb(0)\) be a clustered configuration on the Poincar\'e section and \(T\) its return time. We call a clustered configuration \(\hat{\Xb}(0)\) or a non-clustered configuration \(\hat{\xb}(0)\) \textbf{event close} to \(\Xb(0)\) if \(d(\Xb(0),\hat{\Xb}(0))\) is small enough that 
    $\Xb(t)$ 
     and 
     $\hat{\Xb}(t)$ or 
     $\hat{\xb}(t)$ 
    have the same order of events, up to multiple identical events by different clusters that were synchronized in \(\Xb(t)\) or multiple identical events by difference cells that were in the same cluster in \(\Xb(t)\), for a time \(\epsilon_0<t<T+\epsilon_1 \), where \(\epsilon_0\geq0\) and \(\epsilon_1\geq0\) are small. The set of such \(\hat{\Xb}\) and \(\hat{\xb}\) is called an \textbf{event small} neighborhood of \(\Xb(0)\). For cells in \(\hat{\xb}(0)\), we also require that each cohort close to a cluster in \(\Xb(0)\) have the same number of cells as that cluster. 
\end{definition}

For example, suppose \(\Xb(t)\) has order of events \(\mathbf{s_1r_1s_2r_2}\). Then a cluster or pair of clusters will cross \(s_1\), call this group the \textit{first cohort}, then another cluster or pair of clusters will cross \(r_1\), call this the \textit{second cohort}. Then a cohort will cross \(s_2\), then one will cross \(r_2\), and then finally one will cross \(1\). If \(\hat{\xb}(0)\) is event-close to \(\Xb(0)\), then \textit{every} cell in its first cohort will cross \(s_1\) before \textit{any} cell in its second cohort crosses \(r_1\). Then every cell in the second cohort will cross \(r_1\) before any cell in the third cohort crosses \(s_2\), and so on.


\subsection{Useful Results About Roots of Polynomials}
\label{sec:roots}

In order to prove asymptotic stability of periodic orbits in later sections, we will need to show that the roots of relevant characteristic polynomials are inside the unit disk in $\mathbb{C}$. 
The following theorem will allow us to bound the roots of certain polynomials.

\begin{theorem} [\citet{Anderson}] \label{thrm:Anderson}
    Let $p_n(z)=\sum_{j=0}^{n}a_jz^j$, $n\geq 1$, be any polynomial with $a_i>0$ for all $0\leq i \leq n$. Setting
    \begin{align*}
        \alpha := \min_{0\leq i < n} \Big\{ \frac{a_i}{a_{i+1}} \Big\} \quad \text{ and } \quad  
          \beta := \max_{0\leq i < n} \Big\{ \frac{a_i}{a_{i+1}} \Big\},
    \end{align*}
then all the zeros of $p_n$ are contained in the annulus $\alpha \leq |z| \leq \beta$. 
\end{theorem}

\begin{corollary}[\cite{Anderson}] \label{cor:AndersonCor1}
If the coefficients of $p$ satisfy $\displaystyle \frac{a_0}{a_1} < \beta$, 
then all zeros of $p_n$ satisfy $|z |< \beta$.
\end{corollary}


\section{The Mode 1:2} \label{sec:12}
In a mode 1:2 \ac{PCS}, 
the population of cells is evenly divided into 2 daughter clusters and 1 mother cluster,
where one of the daughter clusters is synchronized with the single mother cluster. 
We will denote the lone daughter cluster as $D_0$, 
and the synchronized daughter and mother clusters as $D_1$ and $M$ respectively.
We will use the set $\{ M = g \}$ as the Poincar\'{e} section.
Thus, if a \ac{PCS} exists, it has the initial condition 
$(D_0(0) , D_1(0) , M(0)) = (0 , g , g)$ 
for some $g$. We know from \cref{thm:existence} there exists such a $g$ that gives a \ac{PCS}, 
but the theorem does not specify the order of events.
In this section, we will examine the order $\bf{r_1 s_1 r_2 s_2 1}$ 
and prove the existence of a 1:2 \ac{PCS} with this order that is asymptotically stable.

\subsection{The Order of Events \texorpdfstring{$\bf{r_1 s_1 r_2 s_2 1}$}{r_1 s_1 r_2 s_2 1}}

In order to achieve the order $\bf{r_1 s_1 r_2 s_2 1}$, 
we must have that $g$ lies between $s_2$ and $r_1$.
First, the pair $(D_1, M)$ reaches the point $r_1$, then $D_0$ reaches $s_1$, 
then the pair $(D_1, M)$ reaches $r_2$, followed by $D_0$ reaching $s_2$ 
and finally the pair $(D_1, M)$ reaches 1. We illustrate the positioning
of initial conditions in relation to parameter values in \cref{fig:12initial}.
 
\begin{figure}[ht]
\begin{center} 
\scalebox{0.7}{
\begin{tikzpicture}
\begin{scope}
\definecolor{qqqqff}{rgb}{0.,0.,1.}
\definecolor{ffqqqq}{rgb}{1.,0.,0.}
\draw[line width=1pt] (0,0) circle (3cm);
\large
\foreach \angle / \label in 
{ 90/$1 \sim 0$,-80 /$g
$ , -90/$r_1$ ,-160/$r_2$ , 70/$s_1$, -30/$s_2$}
{
\draw[line width=2pt] (\angle:3.25cm) -- (\angle:2.75cm);
\draw (\angle:3.9cm) node{\textsf{\label}};
}

\shade[ball color=qqqqff] (-80:2.65) circle (3.5mm);
\shade[ball color=ffqqqq] (-80:3.35) circle (3.5mm);
\shade[ball color=qqqqff] (90:2.65) circle (3.5mm);
\draw[thick][red,->] (-2,4) arc (110:70:6cm) ;
\end{scope}
\end{tikzpicture}
}
\end{center}
\caption{A diagram showing the initial condition for the 1:2 \ac{PCS} with parameter choices that reflect the order $\bf{r_1 s_1 r_2 s_2 1}$. 
Blue spheres represent daughter clusters while the red sphere represents a mother cluster.}
\label{fig:12initial}
\end{figure}
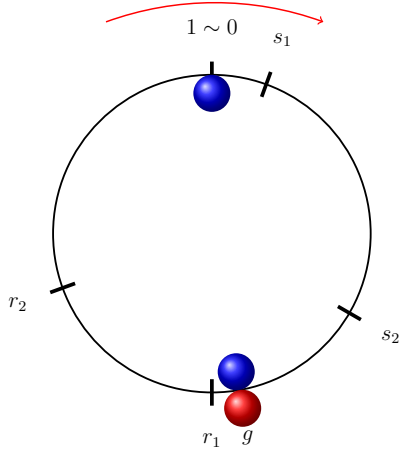


We will now give an explicit formula for $g$ where $s_2 < g < r_1$  given a set of parameters $\{ s_1, s_2, r_1, r_2 \}$.
Throughout this section, we will use the abbreviation $\alpha = 1 + \rho (1/3)$,
which represents the speed of progression of cells in region $R$ when the $1/3$ of the total cells from the lone daughter cluster are currently in $S$. Otherwise, all clusters will progress with rate 1.
\begin{lemma}\label{lem:12g}
    If a mode 1:2 \ac{PCS} follows the order $\bf{r_1 s_1 r_2 s_2 1}$ with positive feedback and $s_2 < g < r_1$, 
    then 
    \begin{align}\label{eqn:12g}
        \boxed{g = 
        \frac{\alpha + (1 - \alpha)( r_2-s_1)}{1 + \alpha}}.
    \end{align}
\end{lemma}

\begin{proof}
Assume $s_2 < g < r_1$ and
\begin{align*}
    D_0(0) &= 0, \\
    D_1(0) &= M(0) = g.
\end{align*}
We will let $\{t_i\}$ denote the times at which events occur for the \ac{PCS}, starting from $t_0 = 0$. 

We begin with the event $r_1$, defining the positions of each cluster for $0 \leq t \leq t_1$ 
where $t_1$ is defined by $D_1(t_1) = M(t_1) = r_1$. 
No clusters are in $R$ or $S$ during this interval, and so all clusters are progressing with rate 1. 
Thus, for $0 \leq t \leq t_1$, we have
\[
\begin{aligned}
    D_0(t) &= t, \\
    D_1(t) &= M(t) = g + t, \quad \text{thus} \quad \boxed{t_1 = r_1 - g}.
\end{aligned}
\]

For $t_1 \leq t \leq t_2$, where $t_2$ is defined by $D_0(t_2) = s_1$, 
no clusters are in $S$, so all clusters are progressing with rate 1. 
Thus, for $t_1 \leq t \leq t_2$, we have
\[
\begin{aligned}
    D_0(t) &= t, \quad \text{thus} \quad \boxed{t_2 = s_1}, \\
    D_1(t) &= M(t) = g + t.
\end{aligned}
\]
Note that when $D_0$ enters $S$, clusters $D_1$ and $M$ are already in region $R$, 
so $D_1$ and $M$ will begin to experience feedback and will progress with rate 
$\alpha$. 

For $t_2 \leq t \leq t_3$, where $t_3$ is defined by $D_1(t_3) = M(t_3) = r_2$,
\[
\begin{aligned}
    D_0(t) &= t, \\
    D_1(t) &= M(t) = \alpha (t - s_1) + s_1 + g , \quad \text{thus} \quad \boxed{t_3 = \frac {r_2 - s_1 - g} {\alpha} + s_1}.
\end{aligned}
\]

For $t_3 \leq t \leq t_4$, where $t_4$ is defined by $D_0(t_4) = s_2$, 
no clusters are in $R$, and so all clusters are progressing with rate 1. 
Thus, for $t_3 \leq t \leq t_4$, we have
\[
\begin{aligned}
    D_0(t) &= t, \quad \text{thus} \quad \boxed{t_4 = s_2} \\
    D_1(t) &= M(t) = t - \frac {r_2 - s_1 - g} {\alpha} - s_1 + r_2.
\end{aligned}
\]

For $t_4 \leq t \leq t^*$ such that $D_1(t^*) = M(t^*) = 1$, 
all clusters are still progressing with rate 1, so we have
\[
\begin{aligned}
    D_0(t) &= t, \\
    D_1(t) &= M(t) = t - \frac {r_2 - s_1 - g} {\alpha} - s_1 + r_2, 
        \quad \text{thus} \quad \boxed{t^* = 1 + s_1 - r_2 + \frac {r_2 - s_1 - g} {\alpha}}.
\end{aligned}
\]

At $t^*$, $D_0$ must be at $g$ while $D_1$ and $M$ will be at 1. Thus,
\[
\begin{aligned}
    D_0(t^*) &= t^* = 1 + s_1 - r_2 + \frac {r_2 - s_1 - g} {\alpha} = g, \\
    D_1(t^*) &= M(t^*) = 1.
\end{aligned}
\]
By solving the equation $D_0(t^*) = g$, we obtain \cref{eqn:12g}.
\end{proof}

Note that in the final step of this proof we found that the time of the final event can be written as $t^* = g$. 
The following shows that the conditions required for the proof of \cref{lem:12g} can be achieved.

\begin{lemma}\label{lem:12anypos}
    For any positive feedback, there exists a non-empty set of parameters \\
    $\{s_1, s_2, r_1, r_2\}$ 
    for which the order $\bf{r_1 s_1 r_2 s_2 1}$ can be followed by a 1:2 \ac{PCS}.
\end{lemma}

\begin{proof}
We need to show that there exist $\{s_1, s_2, r_1, r_2\}$ so that:
\begin{enumerate}
    \item The parameters are in the order we assumed, 
    $0 < s_1 < s_2 < r_1 < r_2 < 1$.

    \item The times of events in the proof of \cref{lem:12g} occur in the order assumed, \\
    $0 < t_1 < t_2 < t_3 < t_4 < t^*$, 
    which means 
    $\displaystyle 0 < r_1 - g < s_1 < \frac {r_2 - s_1 - g} {\alpha} + s_1 < s_2 < g$.

    \item The $g$ in \cref{lem:12g} is after $S$ and before $R$, 
    as we assumed $s_2 < g < r_1$.
\end{enumerate}

First select 
$\displaystyle s_1 \in (0 , \frac{1}{1 + \alpha}) \subset (0,\frac{1}{2})$
and then select
$\displaystyle r_2 \in (\frac{1}{2} + s_1 , 1)$.
Thus we have
$\displaystyle \frac{1}{2} < r_2 - s_1 < 1$.
We will substitute $\delta = r_2 - s_1$ as the parameters are assumed fixed, and we can now write $g$ as a function of $\delta$, yielding
\[
g(\delta) =
\frac{\alpha + \delta (1 - \alpha)}{1 + \alpha}.
\]
As the quantity 
$\displaystyle 1 - \alpha$ 
is negative, we can see that $g(\delta)$ is affine decreasing. 
Thus, we can find the maximum and minimum values for $g$ by evaluating the limits as $\delta$ approaches its lower bound (maximum) and it's upper bound (minimum):
\begin{align*}
    \lim_{\delta \to \frac{1}{2}^+} g(\delta) &= 
    \frac{\frac{1}{2} + \frac{1}{2} \alpha}{1 + \alpha}  =
    \frac{1}{2},
    \\
    \lim_{\delta \to 1^-} g(\delta) &= 
    \frac{1}{1 + \alpha}.
\end{align*}
Therefore, the $g$ in \cref{eqn:12g} satisfies
\[
s_1 < \frac{1}{1 + \alpha} < g < \frac{1}{2} < r_2.
\]

Note that we have thus satisfied $0 < s_1 < g < r_2 < 1$. 
We will now define $s_2$ and $r_1$ so that the times of events occur in our desired order, 
then we will return to this inequality to show that the parameters are in the order we previously assumed.

In order to verify that the times of events occur in our desired order, we will examine and verify $\displaystyle 0 < r_1 - g < s_1 < \frac {r_2 - s_1 - g} {\alpha} + s_1 < s_2 < g$.
First we select 
$\displaystyle r_1 \in (g , s_1 + g) 
$. 
Selecting $r_1$ in this way verifies the first inequality $0 < r_1 - g < s_1$. 
The following inequality, 
$\displaystyle s_1 < \frac{r_2 - s_1 - g}{\alpha} + s_1$ 
is easily verifiable as $r_2 - s_1 > \frac{1}{2} > g$.
Last we wish to select
$\displaystyle s_2 \in (\frac {r_2 - s_1 - g} {\alpha} + s_1 , g) \subset (0, \frac{1}{2})$, which would automatically verify the final inequality, but in order for this choice to be valid we must have that $\displaystyle \frac {r_2 - s_1 - g} {\alpha} + s_1 < g$. Examining this inequality and substituting $\delta = r_2 - s_1$, we can rewrite it as
\begin{align*}
    \delta - g + \alpha s_1 &< \alpha g, \\[4pt]
    \delta + \alpha s_1 &< g(1 + \alpha), \\[4pt]
    \frac{\delta + \alpha s_1}{1 + \alpha} & < g = \frac{\alpha + \delta (1 - \alpha)}{1 + \alpha}
    = \frac{\delta + \alpha (1 - \delta)}{1 + \alpha}
    = \frac{\delta + \alpha s_1 + \alpha (1 - r_2)}{1 + \alpha}.
\end{align*}
Because $1 - r_2 > 0$, the interval $ \displaystyle (\frac {r_2 - s_1 - g} {\alpha} + s_1 , g)$ is non-empty, and so we can select $s_2$ in the desired way to verify the final inequality.

Now that we have selected values for each parameter so that the time of each event occurs in the assumed order, we must check that the parameters themselves are in the assumed order. 
We already know that $0 < s_1 < g < r_2 < 1$, and it is easy to verify that $s_2$ lies between $s_1$ and $g$. 
We can also see clearly that $r_1 > g$, 
and with slightly more analysis we see that because $s_1 < r_2 - g$, we know that $r_1 < r_2$. 
Thus, we have guaranteed that 
$0 < s_1 < s_2 < g < r_1 < r_2 < 1$, 
proving both the first and third conditions.
\end{proof}


\subsection{Derivation of Maps for Mode 1:2 Solutions with Order \texorpdfstring{$\bf{r_1 s_1 r_2 s_2 1}$}{r_1 s_1 r_2 s_2 1}} 

Suppose $\Xb(t) = (D_0(t), D_1(t), M(t))$ is a 1:2 clustered solution that is event close to a 1:2 PCS. 
An initial condition on the Poincaré section has coordinates
\[
\begin{aligned}
D_0(0) &= d_1, 
\\
D_1(0) &= d_2, 
\\
M(0) &= g,
\end{aligned}
\]
where $d_0 \in (-\delta_0, \delta_0)$ and $d_1 \in (g-\delta_1, g+\delta_1)$. 
By $(-\delta_0, \delta_0)$ we mean a $\delta$ neighborhood of $0 \sim 1$ on the circle.

Let $t_{\bf f}^*$ be defined as the smallest positive time for which $M(t_{\bf f}^*) = 1$. 
We will define the map $F_{\bf f}$ for the $\bf C_f$ model as:
\[
F_{\bf f}(d_0, d_1) = (D_1(t_{\bf f}^*), D_0(t_{\bf f}^*)).
\]
Here $D_0$ and $D_1$ remain daughters and $D_1(t_{\bf f}^*)$ is in a neighborhood of the initial condition $d_0$ of $D_0$ while $D_0(t_{\bf f}^*)$ is in a neighborhood of $d_1$.

Let $t_{\bf s}^*$ be defined by $D_1(t_{\bf s}^*) = 1$. 
Then the map $F_{\bf s}$ map for $\bf C_s$ model is defined as:
\[
F_{\bf s}(d_0, d_1) = (M(t_{\bf s}^*), D_0(t_{\bf s}^*)).
\]
Here $M$ becomes a daughter in a neighborhood of $d_0$ while $D_1$ has become a mother and is at $g$.

Before beginning our derivation, we note that we cannot guarantee the order of the initial conditions $d_1$ and $g$.
Thus, we must consider both cases, $d_1 < g$ and $d_1 > g$, but we will see that both cases give the same result.
In either case, each of the events $r_1$ and $r_2$ occur twice by $D_1(t)$ and $M(t)$, while $s_1$ and $s_2$ occur just once by $D_0(t)$.

In our calculations, we will find the time values when clusters change speed, as other events do not change the equations of the positions of the clusters. 
Finally, we will calculate the final time step where either $M$ hits 1 (Fixed model), or $D_1$ hits 1 (Switching model), denoted as $t_{\bf f}^*$ or $t_{\bf s}^*$ respectively. 
Observe \cref{fig:12lineplot} for a visual interpretation of the derivation process. Furthermore, note that while we can see the clusters get closer than their initial conditions after $M$ and $D_1$ have passed 1 and moved to either 0 or $g$, this graph does not provide proof of this phenomena, and the full derivation of the eigenvalues of the map is necessary.

We can additionally observe this phenomena in our simulations by plotting the time evolution of a $1:2$ solution as seen in \cref{fig:12heatmap}. It is clear that the simulation supports such a visualization of the clustering process, though \cref{fig:12lineplot} exaggerates some elements to foster better comprehension of the events taking place.

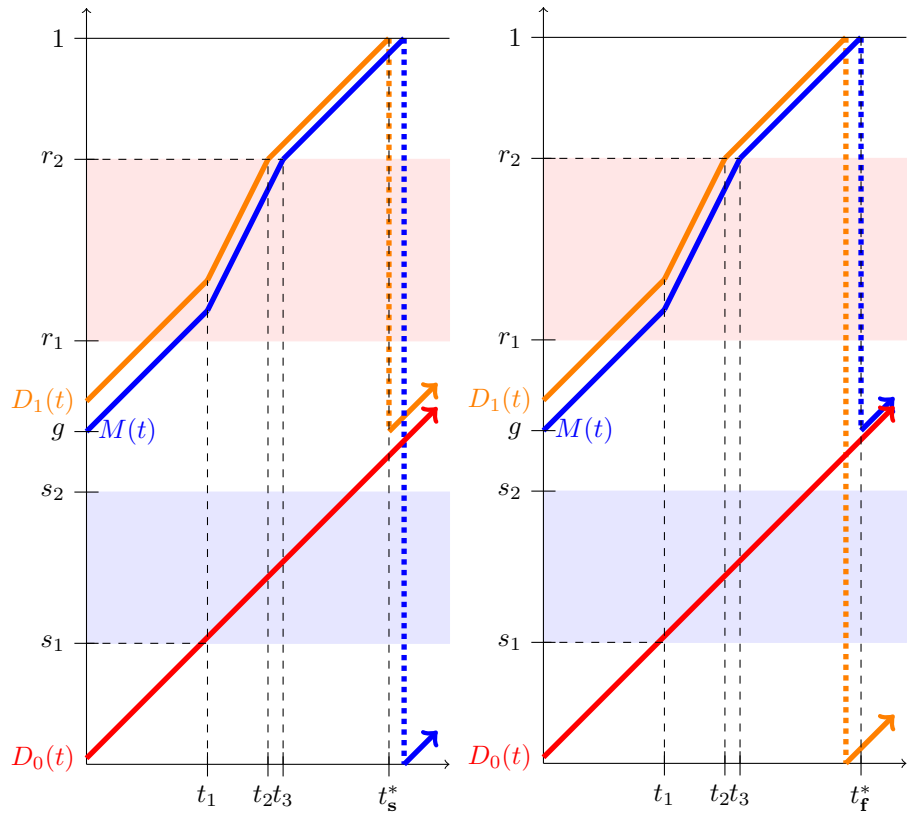
\begin{figure}[h!]
\begin{tikzpicture}[scale=0.8]
\scalebox{1}{

\filldraw[color=blue!10, fill=blue!10] (0,2) rectangle (6,4.5);
\filldraw[color=red!10, fill=red!10] (0,7) rectangle (6,10);

\draw[->] (0,0) -- (6,0);
\draw[->] (0,0) -- (0,12.5);

\draw[line width = 2pt, orange] (0,6) node [anchor=east]{\(D_1(t)\)} -- (2,8);
\draw[line width = 2pt, orange] (2,8) -- (3,10);
\draw[line width = 2pt, orange] (3,10) -- (5,12);
\draw[line width = 2pt, dotted, orange] (5,12) -- (5,5.5);
\draw[->, line width = 2pt, orange] (5,5.5) -- (5.8,6.3);

\draw[line width = 2pt, blue] (0,5.5) node [anchor=west]{\(M(t)\)} -- (2,7.5);
\draw[line width = 2pt, blue] (2,7.5) -- (3.25,10);
\draw[line width = 2pt, blue] (3.25,10) -- (5.25,12);
\draw[line width = 2pt, dotted, blue] (5.25,12) -- (5.25,0);
\draw[->, line width = 2pt, blue] (5.25,0) -- (5.8,0.55);

\draw[->, line width = 2pt, red] (0,0.1) node [anchor=east]{\(D_0(t)\)} -- (5.8,5.9);

\draw[-] (-.2,5.5) node [anchor=east]{\(g\)} -- (.2,5.5);

\draw[-] (-.2,10) node [anchor=east]{\(r_2\)} -- (.2,10);
\draw[-, dashed] (0,10) -- (3.25,10);

\draw[-] (-.2,7) node [anchor=east]{\(r_1\)} -- (.2,7);

\draw[-] (-.2,2) node [anchor=east]{\(s_1\)} -- (.2,2);
\draw[-, dashed] (0,2) -- (2,2);
\draw[-, dashed] (2,0) -- (2,8);

\draw[-] (-.2,4.5) node [anchor=east]{\(s_2\)} -- (.2,4.5);

\draw[-] (-.2,12) node [anchor=east]{1} -- (6,12);

\draw[-] (2,-.2) node [anchor=north]{\(t_1\)} -- (2,.2);

\draw[-] (3,-.2) node [anchor=north]{\(t_2 \ \)} -- (3,.2);
\draw[-, dashed] (3,0) -- (3,10);

\draw[-] (3.25,-.2) node [anchor=north]{\(t_3\)} -- (3.25,.2);
\draw[-, dashed] (3.25,0) -- (3.25,10);

\draw[-] (5,-.2) node [anchor=north]{\(t_{\bf s}^*\)} -- (5,.2);
\draw[-, dashed] (5,0) -- (5,12);
}
\end{tikzpicture}
\begin{tikzpicture}[scale=0.8]
\scalebox{1}{

\filldraw[color=blue!10, fill=blue!10] (0,2) rectangle (6,4.5);
\filldraw[color=red!10, fill=red!10] (0,7) rectangle (6,10);

\draw[->] (0,0) -- (6,0);
\draw[->] (0,0) -- (0,12.5);

\draw[line width = 2pt, orange] (0,6) node [anchor=east]{\(D_1(t)\)} -- (2,8);
\draw[line width = 2pt, orange] (2,8) -- (3,10);
\draw[line width = 2pt, orange] (3,10) -- (5,12);
\draw[line width = 2pt, dotted, orange] (5,12) -- (5,0);
\draw[->, line width = 2pt, orange] (5,0) -- (5.8,0.8);

\draw[line width = 2pt, blue] (0,5.5) node [anchor=west]{\(M(t)\)} -- (2,7.5);
\draw[line width = 2pt, blue] (2,7.5) -- (3.25,10);
\draw[line width = 2pt, blue] (3.25,10) -- (5.25,12);
\draw[line width = 2pt, dotted, blue] (5.25,12) -- (5.25,5.5);
\draw[->, line width = 2pt, blue] (5.25,5.5) -- (5.8,6.05);

\draw[->, line width = 2pt, red] (0,0.1) node [anchor=east]{\(D_0(t)\)} -- (5.8,5.9);

\draw[-] (-.2,5.5) node [anchor=east]{\(g\)} -- (.2,5.5);

\draw[-] (-.2,10) node [anchor=east]{\(r_2\)} -- (.2,10);
\draw[-, dashed] (0,10) -- (3.25,10);

\draw[-] (-.2,7) node [anchor=east]{\(r_1\)} -- (.2,7);

\draw[-] (-.2,2) node [anchor=east]{\(s_1\)} -- (.2,2);
\draw[-, dashed] (0,2) -- (2,2);
\draw[-, dashed] (2,0) -- (2,8);

\draw[-] (-.2,4.5) node [anchor=east]{\(s_2\)} -- (.2,4.5);

\draw[-] (-.2,12) node [anchor=east]{1} -- (6,12);

\draw[-] (2,-.2) node [anchor=north]{\(t_1\)} -- (2,.2);

\draw[-] (3,-.2) node [anchor=north]{\(t_2 \ \)} -- (3,.2);
\draw[-, dashed] (3,0) -- (3,10);

\draw[-] (3.25,-.2) node [anchor=north]{\(t_3\)} -- (3.25,.2);
\draw[-, dashed] (3.25,0) -- (3.25,10);

\draw[-] (5.25,-.2) node [anchor=north]{\(t_{\bf f}^*\)} -- (5.25,.2);
\draw[-, dashed] (5.25,0) -- (5.25,12);
}
\end{tikzpicture}
\caption{Visualizing the positions and movement of the clusters $D_0$, $D_1$ and $M$ under the switching (left) and fixed (right) clustered models. When a cluster is in the blue shaded region which denotes $S$, it signals to the clusters in the red shaded region $R$ to apply positive feedback and progress at rate $\alpha > 1$. The time axis notes all time values when cluster speeds change, which are calculated in the derivation of the maps. The distance between $D_1(0)$ and $M(0)$ has been exaggerated to better view the phenomena of the distance between the clusters decreasing under positive feedback with the order $\bf r_1 s_1 r_2 s_2 1$. One can also notice the movements of $D_1$ and $M$ after they pass 1, where in the fixed model $M$ goes to $g$ and stays a mother cluster, while in the switching model $D_1$ goes to $g$ and becomes a mother cluster.}
\label{fig:12lineplot}
\end{figure}

\begin{figure}[ht] 
  \includegraphics[width=\textwidth]{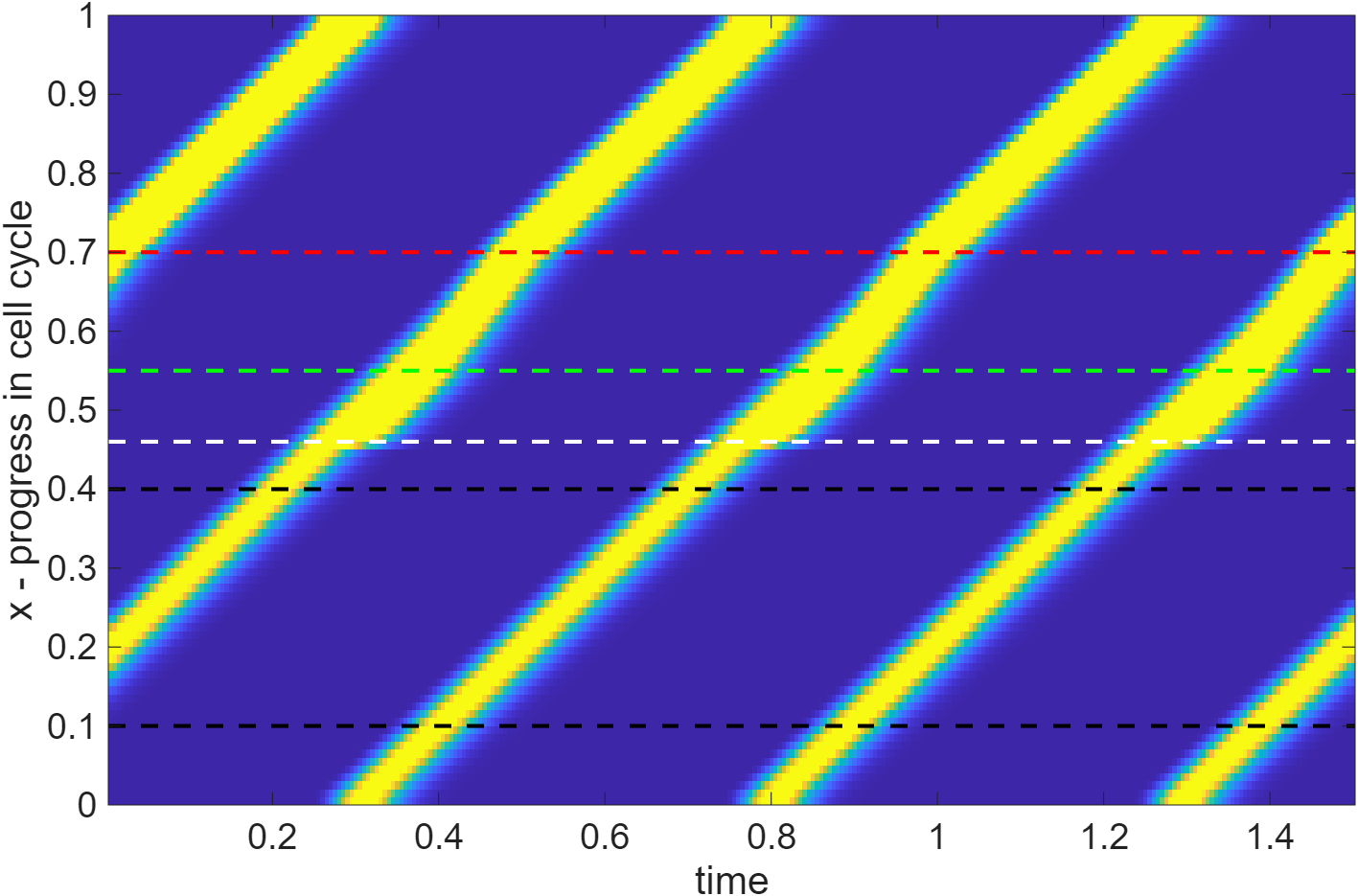}
\caption{Density plot showing the time evolution of a $1:2$ not-perfectly clustered solution of  \cref{basiccell,I} 
with assumption $\bf c_a$ has formed. 
This solution is similar to that seen in the left panel of \cref{fig:12plot}, but with $\rho(I) = 1$ and $g=0.46$. 
These changes exaggerate the change in slope within $R$ and the dispersion of clusters near $g$.
The boundaries of the $S$ region, 
at $s_1 = 0.1$ and $s_2 = 0.4$ are displayed as black dashed lines. 
The green dashed line at $x = 0.55$ is $r_1$, 
while $r_2 = 0.7$ is shown as a red line 
and $g = 0.46$ is shown as a white line.
The order of events for this solution is $\mathbf{r_1 s_1 r_2 s_2 1}$. 
At $g$, a mother cluster joins the daughter cluster, 
but slightly behind the daughter cluster 
and this causes the clusters to disperse.
When a cohort of mother/daughter cells crosses $r_1$, 
there is are no cells in $S$. 
While the mother/daughter cohort is still in $R$, a daughter cluster enters $S$ and this causes all the cells in the $R$ to speed up uniformly. 
When the mother/daughter pair leaves $R$ they slow down, thus compressing.
}
\label{fig:12heatmap}
\end{figure}


Let $t_1$ denote the time at which $D_0(t_1) = s_1$. 
Then the interval $0 \leq t \leq t_1$ 
includes the first event of the pair $(D_1 , M)$ passing $r_1$, but $D_0$ is not in $S$ yet, so all clusters are progressing with rate 1. 
Thus, for $0 \leq t \leq t_1$, we have
\[
\begin{aligned}
    D_0(t) &= d_0 + t,
    \quad \text{thus} \quad
    \boxed{t_1 = s_1 - d_0},
    \\
    D_1(t) &= d_1 + t,
    \\
    M(t) &= g + t.
\end{aligned}
\]

For $t_1 \leq t \leq t_2$, where $t_2$ is defined by $M(t_2) = r_2$, 
$D_0$ will still be in $S$, while $D_1$ and $M$ will be in $R$ progressing at a rate of 
$\alpha$. 
Thus, for $t_1 \leq t \leq t_2$,
\[
\begin{aligned}
    D_0(t) &= d_0 + t,
    \\
    D_1(t) &= \alpha (t - s_1 + d_0) + s_1 - d_0 + d_1,
    \\
    M(t) &= \alpha (t - s_1 + d_0) + s_1 - d_0 + g,
    \quad \text{thus} \quad
    \boxed{t_2 = \frac{r_2 - s_1 + d_0 - g}{\alpha} + s_1 - d_0}.
\end{aligned}
\]

For $t_2 \leq t \leq t_3$, where $t_3$ is defined by $D_1(t_3) = r_2$, 
$D_0$ will still be in $S$ progressing at rate 1, $M$ will have left $R$ and will be progressing at rate 1, and $D_1$ will be in $R$ progressing at a rate of 
$\alpha$. 
Thus, for $t_2 \leq t \leq t_3$,
\[
\begin{aligned}
    D_0(t) &= d_0 + t,
    \\
    D_1(t) &= \alpha(t - s_1 + d_0) + s_1 - d_0 + d_1,
    \quad \text{thus} \quad
    \boxed{t_3=\frac{r_2 - s_1 + d_0 - d_1}{\alpha} + s_1 - d_0},
    \\
    M(t) &= (t - \frac{r_2 - s_1 + d_0 - g}{\alpha} - s_1 + d_0) + r_2.
\end{aligned}
\]

For $t_3 \leq t \leq t_4$, where $t_4$ is defined by $M(t_4) = 1$, 
$D_0(t)$ will have left $S$ during this interval but will still progress at rate 1, and both $M$ and $D_1$ will have already left $R$ and will progress at rate 1. 
Thus, for $t_3 \leq t \leq t_4$,
\begin{align}
\begin{split}
    D_0(t) &= d_0 + t,
    \\
    D_1(t) &= (t - \frac{r_2 - s_1 + d_0 - d_1}{\alpha} - s_1 + d_0) + r_2,
    \\
    M(t) &= (t - \frac{r_2 - s_1 + d_0 - g}{\alpha} - s_1 + d_0) + r_2,
    \\
    \quad &\text{thus} \quad
    \boxed{t_4 = 1 - r_2 + \frac{r_2 - s_1 + d_0 - g}{\alpha} + s_1 - d_0}
\label{eqn:12t4-}
\end{split}
\end{align}

We note here that none of our event time steps so far have had any dependence on the sign of $d_1 - g$, and thus these times are not dependent on our two possible cases of initial conditions.

\textbf{Now considering the fixed clustered model $\bf C_f$}, $M(t_4)$ has now arrived at 1 and goes immediately to $g$, 
so the solution is on the global Poincar\'{e} section $\{M = g\}$, 
but the positions of $D_0$ and $D_1$ have switched, so it is not in a neighborhood of the initial conditions. 
We have completed the map $F_{\bf f}$ and $t_{\bf f}^* = t_4$. 
Thus by evaluating the positions of the clusters at time $t_{\bf f}^*$, we find
\[
\begin{aligned}
    D_0(t_{\bf f}^*) &= 1 - r_2 + s_1 + \frac{r_2 - s_1 + d_0 - g}{\alpha}, \\
    D_1(t_{\bf f}^*) &= 1 + \frac{d_1 - g}{\alpha} \quad \text{(mod 1)},\\
    M(t_{\bf f}^*) &= 1 \rightarrow g.
\end{aligned}
\]
We include that $D_1(t_{\bf f}^*)$ is mod 1 to account for each case of $d_1 - g$, as depending on the order of these initial conditions, it may be located after or before $1 \sim 0$. We will use the same convention later when we find when $D_1$ crosses 1.

Since the locations of $D_0$ and $D_1$ are permuted, $F_{\bf f}$ is given by 
\[ 
F_{\bf f}:(d_0, d_1)\longrightarrow (D_1(t_{\bf f}^*), D_0(t_{\bf f}^*)). 
\]
Thus by taking the Jacobian,
\begin{equation}
    DF_{\bf f} = \displaystyle
    \begin{pmatrix}
        0 & \dfrac{1}{\alpha} \\[8pt]
        \dfrac{1}{\alpha} & 0 
    \end{pmatrix}.
    \label{matrix:12DFf}
\end{equation}

\textbf{For the switching clustered model $\bf C_s$} when assuming $d_1 < g$, 
we have that $t_{\bf s}^*$ would be $t_5$ where $D_1(t_5) = 1$, 
thus we recall equations \cref{eqn:12t4-} and compute the solution for another time interval. 
When $t_4 \leq t \leq t_5 = t_{\bf s}^*$, everything is progressing at rate 1. So,
\[
\begin{aligned}
    D_0(t) &= t + d_0,
    \\
    D_1(t) &= (t - \frac{r_2 - s_1 + d_0 - d_1}{\alpha} - s_1 + d_0) + r_2,
    \\
    \quad &\text{thus} \quad
    \boxed{t_{\bf s}^* = 1 - r_2 + \frac{r_2 - s_1 + d_0 - d_1}{\alpha} + s_1 - d_0}
    \\
    M(t) &= (t - 1 + r_2 - \frac{r_2 - s_1 + d_0 - g}{\alpha} - s_1 + d_0).
\end{aligned}
\]
Thus, by evaluating at $t_{\bf s}^*$, we find
\[
\begin{aligned}
    D_0(t_{\bf s}^*) &= 1 - r_2 + s_1 + \frac{r_2 - s_1 + d_0 - d_1}{\alpha},
    \\
    D_1(t_{\bf s}^*) &= 1 \rightarrow g,
    \\
    M(t_{\bf s}^*) &= 1 + \frac{g - d_1}{\alpha} \quad \text{(mod 1)},
\end{aligned}
\]
which defines the map 
$F_{\bf s} : (d_0 , d_1) \longrightarrow (M(t_{\bf s}^*) , D_0(t_{\bf s}^*))$
with derivative
\begin{equation}
    DF_{\bf s} =
    \begin{pmatrix}
        0 & -\dfrac{1}{\alpha} \\[8pt]
        \dfrac{1}{\alpha} & -\dfrac{1}{\alpha}
    \end{pmatrix}.
    \label{matrix:12DFs}
\end{equation}


\subsection{Stability for a 1:2 PCS with order \texorpdfstring{$\bf{r_1 s_1 r_2 s_2 1}$}{r_1 s_1 r_2 s_2 1}}

\begin{lemma}
\label{lem:12stab}
    For the 1:2 \ac{PCS} with the order $\mathbf{r_1 s_1 r_2 s_2 1}$ and under positive feedback with $\alpha > 1$ in both clustered models $\bf C_f$ and $\bf C_s$, 
    all of the eigenvalues of the Jacobian matrices corresponding to the respective Poincaré maps are in the interior of the unit disk.
\end{lemma}

\begin{proof}
    For the model $\bf C_f$, we find that the characteristic polynomial of \cref{matrix:12DFf} is
    \[
        f(\lambda) = \lambda^2 - \frac{1}{\alpha^2}.
    \]
    The roots of this polynomial,
    \[
        |\lambda| = \left| \frac{1}{\alpha} \right| < 1, \quad \forall \ \alpha > 1.
    \]

    For the model $\bf C_s$, we find that the characteristic polynomial of \cref{matrix:12DFs} is
    \[
        s(\lambda) = \lambda^2 + \frac{1}{\alpha} \lambda + \frac{1}{\alpha^2}
    \]
    with roots,
    \[
        \lambda = \frac{-1 \pm i \sqrt{3}}{2 \alpha}.
    \]
    The moduli of these roots,
    \[
        |\lambda| = \sqrt{ \frac{1}{4\alpha^2} + \frac{3}{4\alpha^2}} = \frac{1}{\alpha} < 1, \quad \forall \ \alpha >1.
    \]
    
\end{proof}

Combining \cref{lem:12g}, \cref{lem:12anypos} and \cref{lem:12stab}, we obtain the following result:

\begin{theorem}
    In the models $\bf C_f$ and $\bf C_s$ under positive feedback, there exists a $1:2$ \ac{PCS} with order $\bf r_1 s_1 r_2 s_2 1$ that is asymptotically stable.
\end{theorem}


\section{The Mode 2:3} \label{sec:23}

In a mode 2:3 \ac{PCS}, 
the population of cells is evenly divided into 3 daughter clusters and 2 mother clusters. 
We will denote these clusters in the same way as in \cref{sec:12}, 
now with the renaming of $M$ to $M_1$ to match its paired daughter cluster $D_1$, 
as well as the addition of a second pair of daughter and mother clusters, $D_2$ and $M_2$. 
As previously, we will use the set $\{ M_1 = g \}$ as the Poincar\'{e} section. 
We define the initial position of $M_2$ as $w$ for the \ac{PCS}. 
Thus, if a \ac{PCS} exists, it has the initial condition 
$(D_0(0) , D_1(0) , D_2(0) , M_1(0) , M_2(0)) = (0 , g , w , g , w)$ 
for some $g$ and $w$ such that $0 < g < w < 1$, as seen in \cref{fig:23_diag}.
Throughout this section on the 2:3 \ac{PCS}, we will have that
$\alpha = 1 + \rho (1/5)$, as the lone daughter cluster, which will be in $S$, 
consists of $1/5$ of the total cells.

As in \cref{sec:12}, 
we will examine the order $\bf{r_1 s_1 r_2 s_2 1}$ 
and prove the existence of a 2:3 \ac{PCS} with this order that is asymptotically stable.


\subsection{The Order of Events \texorpdfstring{$\bf{r_1 s_1 r_2 s_2 1}$}{r_1 s_1 r_2 s_2 1}}

\begin{figure}[h!]
\begin{center}  
\scalebox{.9}{
\begin{tikzpicture}[scale=0.85]
\begin{scope} [xshift=2cm]
\definecolor{qqqqff}{rgb}{0.,0.,1.}
\definecolor{ffqqqq}{rgb}{1.,0.,0.}
\draw (0,0) circle (1.25cm);
\foreach \angle / \label in 
{ 90/$1 \sim 0$,-20 /$g$ , 200/$w$, -40/$r_1$ ,-80/$r_2$ , 60/$s_1$, 0/$s_2$}
{
\draw[line width=1pt] (\angle:1.5cm) -- (\angle:1.25cm);
\draw (\angle:1.9cm) node{\textsf{\label}};
}
\foreach \angle in {-60}
\draw[line width=1pt] (\angle:1.6cm) -- (\angle:1.6cm);

\shade[ball color=qqqqff] (-10:1.1) circle (1.5mm);
\shade[ball color=ffqqqq] (210:1.4) circle (1.5mm);
\shade[ball color=qqqqff] (210:1.1) circle (1.5mm);
\shade[ball color=ffqqqq] (100:1.4) circle (1.5mm);
\shade[ball color=qqqqff] (100:1.1) circle (1.5mm);
\draw[thick][red,->] (-1.75,2) arc (120:60:3.5cm) ;
\end{scope}

\begin{scope} [xshift=-5.5cm]
\definecolor{qqqqff}{rgb}{0.,0.,1.}
\definecolor{ffqqqq}{rgb}{1.,0.,0.}
\draw (0,0) circle (1.25cm);
\foreach \angle / \label in 
{ 90/$1 \sim 0$,-20 /$g$ , 200/$w$, -40/$r_1$ ,-80/$r_2$ , 60/$s_1$, 0/$s_2$}
{
\draw[line width=1pt] (\angle:1.5cm) -- (\angle:1.25cm);
\draw (\angle:1.9cm) node{\textsf{\label}};
}
\foreach \angle in {-60}
\draw[line width=1pt] (\angle:1.6cm) -- (\angle:1.6cm);
\shade[ball color=ffqqqq] (-20:1.4) circle (1.5mm);
\shade[ball color=ffqqqq] (200:1.4) circle (1.5mm);
\shade[ball color=qqqqff] (-20:1.1) circle (1.5mm);
\shade[ball color=qqqqff] (90:1.1) circle (1.5mm);
\shade[ball color=qqqqff] (200:1.1) circle (1.5mm);
\draw[thick][red,->] (-1.75,2) arc (120:60:3.5cm) ;
\end{scope}

\draw[thick][blue,->] (-3,0) --node[above]{$\bf{C}_{\bf f}$,$\bf{C}_{\bf s}$ ~ }(0,0);

\begin{scope} [xshift=-2cm,yshift=-5.9cm]
\definecolor{qqqqff}{rgb}{0.,0.,1.}
\definecolor{ffqqqq}{rgb}{1.,0.,0.}
\draw (0,0) circle (1.25cm);
\foreach \angle / \label in 
{ 90/$1 \sim 0$,-20 /$g$ , 200/$w$, -40/$r_1$ ,-80/$r_2$ , 60/$s_1$, 0/$s_2$}
{
\draw[line width=1pt] (\angle:1.5cm) -- (\angle:1.25cm);
\draw (\angle:1.9cm) node{\textsf{\label}};
}
\foreach \angle in {-60}
\draw[line width=1pt] (\angle:1.6cm) -- (\angle:1.6cm);
\shade[ball color=ffqqqq] (-20:1.4) circle (1.5mm);
\shade[ball color=ffqqqq] (200:1.4) circle (1.5mm);
\shade[ball color=qqqqff] (-20:1.1) circle (1.5mm);
\shade[ball color=qqqqff] (90:1.1) circle (1.5mm);
\shade[ball color=qqqqff] (200:1.1) circle (1.5mm);

\draw[thick][red,->] (-1.75,2) arc (120:60:3.5cm) ;
\end{scope}
\draw[thick][blue,->] (3.5,-1.5) --node[above]{ ~ $\bf{C}_{\bf s}$}(5.2,-3.2);
\draw[thick][blue,->] (.5,-1.5) --node[above]{$\bf{C}_{\bf f}$ ~ }(-1.2,-3.1);

\begin{scope}[xshift=6cm,yshift=-5.9cm]
\definecolor{qqqqff}{rgb}{0.,0.,1.}
\definecolor{ffqqqq}{rgb}{1.,0.,0.}
\draw (0,0) circle (1.25cm);
\foreach \angle / \label in 
{ 90/$1 \sim 0$,-20 /$g$ , 200/$w$, -40/$r_1$ ,-80/$r_2$ , 60/$s_1$, 0/$s_2$}
{
\draw[line width=1pt] (\angle:1.5cm) -- (\angle:1.25cm);
\draw (\angle:1.9cm) node{\textsf{\label}};
}
\foreach \angle in {-60}
\draw[line width=1pt] (\angle:1.6cm) -- (\angle:1.6cm);

\shade[ball color=ffqqqq] (-20:1.1) circle (1.5mm);
\shade[ball color=qqqqff] (-20:1.4) circle (1.5mm);
\shade[ball color=qqqqff] (200:1.1) circle (1.5mm);
\shade[ball color=qqqqff] (90:1.4) circle (1.5mm);
\shade[ball color=ffqqqq] (200:1.4) circle (1.5mm);
\draw[thick][red,->] (-1.75,2) arc (120:60:3.5cm);
\end{scope}
\end{tikzpicture}
}
\captionof {figure}{A diagram of dynamics at  division for a 2:3 \ac{PCS} with both clustered models. The upper left diagram represents the initial conditions of the 2:3 \ac{PCS} with parameter choices reflecting the order of events $\bf{r_1s_1r_2s_21}$.
The red spheres represent clusters that start as mothers and blue spheres represent clusters that begin as daughters.
Outer spheres represent clusters that are currently mothers and inner spheres represent current daughter clusters. 
For the fixed model, the cells that began as mothers remain as such and go to $g$ at the time of division. 
For the switching model, the cells that began as daughters become mothers and travel to $g$ when they reach $1$.}
\label{fig:23_diag}
\end{center}
\end{figure}
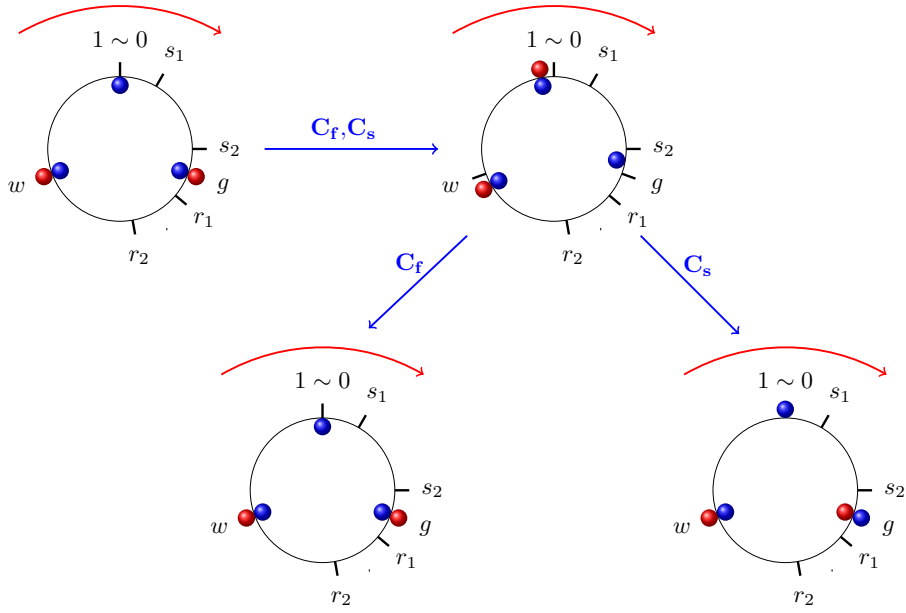

For the order $\bf{r_1 s_1 r_2 s_2 1}$ with $s_2 < g < r_1$ and $r_2 < w < 1$, 
the point $r_1$ is reached by the pair $(D_1, M_1)$, 
then $D_0$ reaches $s_1$, 
then the pair $(D_1, M_1)$ reaches $r_2$, 
followed by $D_0$ reaching $s_2$ and finally the pair $(D_2, M_2)$ reaching 1. 
We will now show that explicit formulas for $g$ where $s_2 < g < r_1$ and $w$ where $r_2 < w < 1$ can be derived given a set of parameters $\{ s_1, s_2, r_1, r_2 \}$.

\begin{lemma}\label{lem:23gw}
If a mode 2:3 \ac{PCS} follows the order $\bf{r_1 s_1 r_2 s_2 1}$ with positive feedback and $s_2 < g < r_1$, then $g$ and $w$  must be equal to
\begin{align}
&\boxed{g = \frac {\alpha + (1 - \alpha) (r_2 - s_1)} {1 + 2 \alpha}},\label{eqn:23g}\\ 
&\boxed{w = 1 - g}.\label{eqn:23w} 
\end{align}
\end{lemma}
This lemma can be proven in the same way as \cref{lem:12g}, where we assume the conditions for a \ac{PCS} have been met and calculate the time at which $D_0$ reaches $g$, then solve for $g$.
The equation for $w$ follows directly from this by substitution. 
The full proof of this lemma can be found in the Supplementary Materials.

The following shows that the conditions required for the proof of \cref{lem:23gw} can be achieved.

\begin{lemma}\label{lem:23existence}
For any positive feedback, there exists a non-empty set of parameters \\
$\{s_1, s_2, r_1, r_2\}$ for which the order $\bf{r_1 s_1 r_2 s_2 1}$ can be followed by a 2:3 \ac{PCS}.
\end{lemma}
To prove this lemma, we must show that there exist $\{s_1, s_2, r_1, r_2\}$ so that:
\smallskip
\begin{enumerate}
    \item The parameters are in the order we assumed $0 < s_1 < s_2 < r_1 < r_2 < 1$.

    \item The times of events in the proof of \cref{lem:23gw} from the Supplementary Materials occur in the order assumed, \\
    $0 < t_1 < t_2 < t_3 < t_4 < t^*$, which means 
    $\displaystyle 0 < r_1 - g < s_1 < \frac {r_2 - s_1 - g} {\alpha} + s_1 < s_2 < 1 - w = g$.

    \item The $g$ in \cref{eqn:23g} is after $S$ and before $R$, as we assumed $s_2 < g < r_1$.

    \item The $w$ in \cref{eqn:23w} is after $R$ and before 1.
\end{enumerate}
\smallskip 
This proof can be completed using algebraic manipulation to provide ranges for each parameter for which a 2:3 \ac{PCS} can be followed. We have included the proof in the Supplementary Materials.


\subsection{Derivation of Maps for Mode 2:3 Solutions with Order \texorpdfstring{$\bf{r_1 s_1 r_2 s_2 1}$}{r_1 s_1 r_2 s_2 1}}

Let $\Xb(t) = (D_0(t), D_1(t), D_2(t), M_1(t), M_2(t))$
be a 2:3 clustered solution near a 2:3 \ac{PCS} such that
\[
\begin{aligned}
    D_0(0) &= d_0,
    \\
    D_1(0) &= d_1,
    & \hspace{3cm}
    M_1(0) &= g,
    \\
    D_2(0) &= d_2,
    &
    M_2(0) &= m_2.
\end{aligned}
\]
where $d_0 \in (-\delta_0 , \delta_0)$, $d_1 \in (g - \delta_1, g + \delta_1)$ 
and $d_2, m_2 \in (w - \delta_2, w + \delta_2)$ 
for small positive $\delta_0, \delta_1, \delta_2$.
We assume $\Xb(0)$ is event close to the \ac{PCS}.
Much like in the derivation for the mode 1:2, 
we will find that the signs of $d_1 - g$ and $d_2 - m_2$ do not impact our results.

We will forgo the explicit derivation which follows the same steps as the calculations done in \cref{sec:12}. 
These calculations can be found in the Supplementary Materials for optional reference.

The resulting $DF$ matrices for the fixed and switching models are as follows:

\begin{equation}
DF_{\bf f} =
\begin{pmatrix}
    0 & 0 & 1 & -1 \\[8pt]
    1 & 0 & 0 & -1 \\[8pt]
    \dfrac{\alpha-1}{\alpha} & \dfrac{1}{\alpha} & 0 & -1 \\[8pt]
\dfrac{\alpha-1}{\alpha} & 0 & 0 & -1
\end{pmatrix}, 
\quad
DF_{\bf s} =
\begin{pmatrix}
    0 & 0 & -1 & 1 \\[8pt]
    1 & 0 & -1 & 0 \\[8pt]
    \dfrac{\alpha-1}{\alpha} & \dfrac{1}{\alpha} & -1 & 0 \\[8pt]
    \dfrac{\alpha-1}{\alpha} & 0 & -1 & 0
\end{pmatrix}.
\label{matrix:23DF}
\end{equation}


\subsection{Stability for a 2:3 PCS with order \texorpdfstring{$\bf{r_1 s_1 r_2 s_2 1}$}{r_1 s_1 r_2 s_2 1}}

\begin{lemma}\label{lem:23stab}
    For the 2:3 \ac{PCS} with the order $\mathbf{r_1 s_1 r_2 s_2 1}$ under positive feedback with $\alpha_{1/5} > 1$ in both clustered models $\bf C_f$ and $\bf C_s$, 
    all of the eigenvalues of the Jacobian matrices corresponding to the respective Poincaré maps are in the interior of the unit disk.
\end{lemma}

\begin{proof}
    From \cref{matrix:23DF} we get the Jacobian of the $F$ map for the fixed model $\bf C_f$, whose characteristic polynomial is 
    \[
        f(\lambda) = \lambda^4+\lambda^3-\frac{\lambda}{\alpha}-\frac{1}{\alpha^2} = (\lambda^2 -\frac{1}{\alpha})(\lambda^2 + \lambda + \frac{1}{\alpha}).
    \]
    The first factor's roots, $\sqrt{1/\alpha}$, are in the unit disc. 
    The second factor has roots
    \[
        \lambda = \frac{-1 \pm \sqrt{1-\frac{4}{\alpha}}}{2}.
    \]
    We now note that $\lambda$ is real when $\alpha \geq 4$ and is complex for $1 < \alpha < 4$. For the complex case, we can write
    \[
     \lambda = \frac{-1}{2} \pm \frac{i}{2} \sqrt{\frac{4}{\alpha} - 1}.
    \]
Taking the modulus, we find
\[
| \lambda | = \sqrt{\frac{1}{4} + \frac{1}{4} \left( \frac{4}{\alpha} - 1 \right)}
=
\sqrt{\frac{1}{\alpha}}
< 1 \quad \forall \ \alpha > 1.
\]
For $\alpha \geq 4$ where $\lambda$ is real, we write
\[
| \lambda | = \frac{1}{2} + \frac{1}{2} \sqrt{1 - \frac{4}{\alpha}},
\]
and since $1 - \frac{4}{\alpha} < 1$ for all $\alpha \geq 4$, we have that $\lambda$ is in the unit disk for all $\alpha > 1$.
    
    In the case of the switching model, $\bf C_s$, we get the characteristic polynomial
    \[
        s(\lambda) = \lambda^4 + \lambda^3 + \frac{1}{\alpha}\lambda^2 + \frac{1}{\alpha}\lambda + \frac{1}{\alpha^2}
    \]
    which has ratios of coefficients of
    \[
        \left\{\frac{a_i}{a_{i+1}}\right\}_{i=0}^3 = \left\{ \frac{1}{\alpha}, 1 , \frac{1}{\alpha}, 1\right\}. 
    \]
    From this, we have
    \[
        a = \min\{a_i/a_{i+1}\} = \frac{1}{\alpha}, \quad \text{and} \quad b = \max\{a_i/a_{i+1}\} = 1.
    \]
    By \cref{thrm:Anderson}, all of the roots to the characteristic equation are in the interval $[a,b]$. 
    However, by \cref{cor:AndersonCor1}, since
    \[
        b = 1 > \frac{a_0}{a_1} = \frac{1}{\alpha},
    \]
    all of the roots of the characteristic equation for the switching model are strictly less than 1 in modulus, as they fall in the interval 
    $[ \, 1 / \alpha \ , 1)$.
\end{proof}

Combining \cref{lem:23gw}, \cref{lem:23existence} and \cref{lem:23stab}, we obtain the following result:

\begin{theorem}
    In the models $\bf C_f$ and $\bf C_s$ under positive feedback, there exists a $2:3$ \ac{PCS} with order $\bf r_1 s_1 r_2 s_2 1$ that is asymptotically stable.
\end{theorem}


\section{Discussion}

In previous studies of cell cycle populations models, 
it had been shown that negative feedback can lead to asymptotically stable temporally clustered solutions, 
under either symmetric or asymmetric division.  
Further, it had been shown that under symmetric division and positive feedback, 
only the synchronized solution can be asymptotically stable.
In this manuscript, we have demonstrated that {\em asymmetric division} and {\em positive feedback}
can lead to asymptotically stable temporal clustering.
In particular, we have proven the existence and asymptotic stability of 
$1:2$ and $2:3$  perfectly clustered solutions (PCSs). 

As in previous work, the {\em order of events}
plays an important role in the stability.
There are certain conditions on the parameters, related to the order of events, 
under which these \ac{PCS}s exist and are stable, 
but we have proven that these conditions can be satisfied.

Asymptotic stability implies {\em structural stability} \cite{perko2013}.  
This means that for an open set of parameter values that contains those that satisfy the conditions, 
we have qualitatively similar temporally clustered solutions that are also asymptotically stable.
Numerical simulations suggest that the existence and stability are indeed robust.

It follows from \cite{Greg} that stability in the clustered models ${\bf C_f}$ and
${\bf C_s}$ imply stability in the cell models ${\bf c_f}$ and ${\bf c_s}$, respectively. 
Note that we have not proven stability in the discontinuous ``Alternating Cell" model ${\bf c_a}$. 
It appears that producing such a proof would be challenging even though numerical simulations clearly suggest it.


\section*{Declarations}

All authors have made substantial intellectual contributions to the study conception, execution, and design of the work. 
All authors have read and approved the final manuscript. 
In addition, the following contributions occurred:
Conceptualization : Todd R. Young;  
Methodology \& Analytical Computations: All; 
Figure Creation: Graham Walther; 
Simulations: Brynley Needham; 
Funding acquisition \& Supervision: Todd R. Young.

{\bf Participating Investigators:}
Bishal Panthee helped with some initial investigations. 
Prof.~Martin J.~Mohlenkamp: 
participated in technical editing of the manuscript,
greatly improving it's quality and structure.

{\bf Conflicts of Interest:} The authors declare there are no conflicts of interest.

{\bf Data $\&$ Code Availability Statement:} The MATLAB script used to generate simulations referenced in this work can be found in the Supplementary Materials.

{\bf Funding Statement:} 
Camden Kilton and Brynley Needham received funding from the Quantitative Biology Institute at Ohio University.


\bibliographystyle{alphaurl_etd}

\end{document}